\documentclass{amsart}
\usepackage{a4wide}
\usepackage{amsmath} 
\usepackage{graphicx}

\usepackage[british]{babel}
\usepackage[utf8]{inputenc}

\usepackage{mathrsfs}

\usepackage{amssymb}

\usepackage{bm}
\usepackage{enumerate}
\usepackage{hyperref}

\usepackage[backend=bibtex,style=alphabetic,sorting=nty,
    sortcites=true,doi=true,url=false,
    giveninits=true,hyperref =true, isbn=false, maxbibnames=99, maxnames=5]{biblatex}
    \AtEveryBibitem{%
  \clearlist{language}%
}
\newtheorem{lemma}{Lemma}[section]

\newtheorem{corollary}{Corollary}[section]
\newtheorem{proposition}{Proposition}[section]
\newtheorem{theorem}{Theorem}[section]
\newtheorem*{theorem*}{Theorem}
\newtheorem{remark}{Remark}[section]

\def\bma{{\bm a}}
\def\bmb{{\bm b}}
\def\bmc{{\bm c}}
\def\bmd{{\bm d}}
\def\bme{{\bm e}}
\def\bmf{{\bm f}}
\def\bmg{{\bm g}}

\def\bmi{{\bm i}}
\def\bmj{{\bm j}}
\def\bmk{{\bm k}}

\def\bmp{{\bm p}}

\def\bmv{{\bm v}}

\def\bmzero{{\bm 0}}
\def\bmone{{\bm 1}}

\def\bmA{{\bm A}}
\def\bmB{{\bm B}}
\def\bmC{{\bm C}}
\def\bmD{{\bm D}}

\def\bmF{{\bm F}}

\def\bmQ{{\bm Q}}

\def\bmalpha{{\bm \alpha}}

\def\bmomega{{\bm \omega}}

\def\bmtau{{\bm \tau}}

\def\bmpartial{{\bm \partial}}
\def\bmnabla{{\bm \nabla}}

\begin{document}

\title[Conformal Einstein field equations - massless Vlasov matter ]{The conformal Einstein field equations with massless Vlasov matter}

\subjclass[2010]{35L04, 35Q75, 35Q76, 53A30}
 \author[J. Joudioux]{{J\'er\'emie} {Joudioux}}
 \address{Max-Planck-Institut f\"ur Gravitationsphysik (Albert-Einstein-Institut), Am M\"uhlenberg 1, 14476 Potsdam, Germany}
 \email{jeremie.joudioux@aei.mpg.de}

 \author[M. Thaller]{{Maximilian} {Thaller}}
 \address{Department of Mathematical Sciences, University of Gothenburg \& Chalmers University of Technology, Chalmers Tvärgata 3, 41296 G\"oteborg, Sweden}
 \email{maxtha@chalmers.se}
\author[J. A. Valiente Kroon]{{Juan} {A.} {Valiente Kroon}}
 \address{School of Mathematical Sciences, Queen Mary University of
   London, Mile End Road, London E1 4NS}
 \email{j.a.valiente-kroon@qmul.ac.uk}

\thanks{ M.T. thanks the A.E.I. for hospitality. J.J. thanks L. Andersson, H. Andr\'easson and H. Friedrich for their support, and friendly discussions.}

\keywords{Einstein conformal field equations, stability problems, symmetric hyperbolic system, Vlasov equation}

\begin{abstract} We prove the stability of de Sitter space-time as a solution to the Einstein- Vlasov system with massless particles. The semi-global stability of Minkowski space-time is also addressed. The proof relies on conformal techniques, namely Friedrich's conformal Einstein field equations. We exploit the conformal invariance of the massless Vlasov equation on the cotangent bundle and adapt Kato's local existence theorem for symmetric hyperbolic systems to prove a long enough time of existence for solutions of the evolution system implied by the Vlasov equation and the conformal Einstein field equations. 
\end{abstract}


\maketitle

\sloppy

\section*{Introduction}

\subsection*{The Einstein-Vlasov system}
Kinetic theory in general relativity is used to model the global behaviour of a space-time when, for instance, galaxy clusters are identified with particles.  The Vlasov matter model is a sub-class of these kinetic models arising when the collision of particles is neglected. This approximation is relevant when collisions either do not affect the motion of the particles or can be neglected ---for instance in the description of low-density gases. In this article we consider ultra-relativistic matter ---that is to say, particles whose velocity is close to the speed of light. A good approximation would then be that the particles, like photons, are massless and, in the context of general relativity, follow the path of freely falling massless particles ---i.e.~null geodesics. Reviews on relativistic kinetic theory can be found in \cite{MR1466508, Andreasson2011}.

In the context of general relativity, ultra-relativistic matter is modelled by the Einstein-Vlasov system. In the following let $(\tilde{\mathcal M}, \tilde{\bm g})$ be a four dimensional Lorentzian manifold and consider the subset of the cotangent bundle $\mathcal{P}$ defined by
\[
\mathcal{P} = \{ (x,\bmp) \in T^\ast M \, : \, \tilde{\bm g}^{-1}_x(\bmp,\bmp) = 0\}.
\]
The matter distribution at $x$ with momentum $\bmp$ is a non-negative function
\[
f: \mathcal{P} \rightarrow \mathbb{R}_+,
\]
satisfying the transport equation
\[
\tilde{\mathcal{L}} f = 0,
\]
where $\tilde{\mathcal{L}}$ is the Liouville vector field ---i.e. the  Hamiltonian vector field of the Hamiltonian function \((x,\bmp) \mapsto -\frac12 \tilde{\bmg}^{-1}_x(\bmp,\bmp) \). The stress-energy tensor associated with this matter model is
\[
\bm T [f]= T_{\alpha\beta}[f](x) \mathbf{d}x^\alpha \otimes \mathbf{d}x^\beta, \qquad  T_{\alpha\beta}[f](x) = 8\pi \int_{\tilde{\mathcal  P}_x} f(x,p)\, p_\alpha p_\beta \, \tilde{\mathrm{dvol}}_{\tilde{\mathcal  P}_x}.
\]
The Einstein-Vlasov system couples the Einstein equations for the metric \(\tilde{g}\) to the transport equation satisfied by \(f\):
\begin{subequations}
\begin{eqnarray}
\bm{Ric}_{\tilde{g}} - \dfrac12 \bm{R}_{\tilde {g}} \tilde {\bm g} + \lambda \tilde{\bm g} &=& \bm T [f],\\
\tilde{\mathcal{L}} f& = &0.
\end{eqnarray}
\end{subequations}

\subsection*{Stability results}
Written in an appropriate coordinates system, the Einstein-Vlasov system is a hyperbolic system of partial differential equations for which the Cauchy problem is well-posed and admits ``ground state'' solutions. These correspond to the absence of gravitational radiation and the absence of matter. The ground states are Minkowski space-time when the cosmological constant vanishes and de Sitter space-time when the cosmological constant is a negative number in our metric signature convention.  Since the Cauchy problem is well-posed, the dynamical stability of such solutions can be discussed. We are proving here the following theorem:
\begin{theorem*} The de Sitter space-time is a dynamically stable solution of the massless Einstein-Vlasov system. In particular, any small enough perturbation of the de Sitter initial data leads to a future lightlike and timelike geodesically complete space-time.
\end{theorem*}
As part of our analysis, we also prove that the curvature of the perturbation and the stress-energy tensor approach zero asymptotically. To that extent, this theorem can be seen as an asymptotic stability result.

This result has already been proven by Ringstr\"om \cite{MR3186493}. The novelty lies here in the approach. The particularity of the Vlasov equation, written on the cotangent bundle, is that solutions are conformally invariant. This is a  consequence of two facts: up to reparametrisation, null geodesics for conformal metrics coincide; and the Vlasov equation implies that the distribution function is constant along null geodesics. This conformal invariance suggests the use of conformal techniques to analyse the stability properties of solutions to the Einstein-Vlasov system. The use of conformal techniques reduces the problem of the global existence of solutions to the Einstein-Vlasov system to the study of the local existence for sufficiently long times of a symmetric hyperbolic problem obtained from a conformal recasting of the Einstein-Vlasov system. The long (enough) time of existence of solutions to this system obtained using the theory developed by Kato \cite{kato}.

Following the same strategy, we also address the problem of the stability of Minkowski space-time. It is well-known that the standard compactification of Minkowski space-time contains a singular point, $i^0$, corresponding to the end-points of the inextendible spacelike geodesics. Hence, to avoid the problem of considering the evolution of the initial data from the point $i^0$, we work with perturbations which, initially, {coincide with the  Schwarzschild solution outside a} compact set ---in the spirit of the glueing results \cite{MR1902228,MR2225517}. The construction of such data is not discussed here. These data evolve to  {coincide with the  Schwarzschild solution outside a} compact set on a hyperboloid ---that is to say, on a spacelike hypersurface asymptote to a Minkowski light-cone, and transverse to the boundary at infinity. From the Cauchy stability results for the Einstein-Vlasov system \cite{AIF_1971__21_3_181_0, MR3186493} it follows that small enough initial perturbations remain small in the evolution. We then prove the following semi-global stability result of Minkowski space:


\begin{theorem*} Consider a hyperboloidal initial data set for the massless Einstein-Vlasov system close enough to the initial data giving rise to the Minkowski space-time. These initial data lead to a future lightlike and timelike geodesically complete space-time, solution to the massless Einstein-Vlasov system.
\end{theorem*}


As before, the curvature and the stress-energy tensor asymptote zero, and this result can consequently be considered as a semi-global asymptotic stability result. Note that, in the absence of estimates for the conformal factor, this result contains no estimates on the behaviour of the curvature, other than the convergence to zero as a power of the conformal factor. Nonetheless, it can be expected that the decay would be stronger than the decay obtained through the vector field method by a factor $(1+|t-r|)^{1/2}$.

Following the work of Dafermos \cite{MR2289606} in spherical symmetry, Taylor has already proven the stability of Minkowski space-time \cite{MR3629140} as a solution to the Einstein-massless Vlasov system. The method relies on energy estimates in the double null gauge. It is interesting to note that in that work the Vlasov matter is compactly supported both in space and momentum, while the stability outside the space support of the matter distribution is handled by the stability and peeling result of \cite{MR1946854}. Nonetheless, the metric data handled in \cite{MR3629140} are not compactly supported, hence, \cite{MR3629140} is a global stability result. A complete stability of Minkowski space-time as a solution to the Einstein-massless Vlasov system result can be found  in \cite{Bigorgne:2020aa}.

\subsection*{Conformal techniques and stability problems}

The problem of stability of solutions in general relativity has been, since the '80s, at the core of many publications. A stepping stone is the work by Christodoulou-Klainerman \cite{MR1316662}, which has been simplified, and reproduced for other matters models (electro-magnetism \cite{MR2531716}, non-linear electro-magnetism \cite{MR3012654}, Klein-Gordon \cite{MR3535896}, massless \cite{MR3629140}, and massive Vlasov \cite{2017arXiv170706079L,2017arXiv170706141F}). Nonetheless, the first and pioneering stability result has been obtained by Friedrich \cite{MR868737} for the de Sitter space-time, by methods of a completely different nature, exploiting the conformal structure of the underlying space-time.

Conformal methods are well adapted to the study of the longtime existence and stability of asymptotically simple space-times. In \cite{MR1131434}, Friedrich's conformal Einstein field equations have been used to study the semi-global stability of the Minkowski space-time and the global stability of the de Sitter space-time under non-vacuum perturbations sourced by the Maxwell and Yang-Mills fields. These results generalise the original vacuum stability results in \cite{MR868737}. The key property that makes the Maxwell and Yang-Mills fields amenable to a treatment using conformal methods is the fact that their energy-momentum tensor has vanishing trace. As a consequence, the conservation equation satisfied by the energy-momentum tensor is conformally invariant. Moreover, the field equations satisfied by the Maxwell and Yang-Mills fields can also be shown to be suitably conformally invariant. Other fields with trace-free energy-momentum tensor are the conformally coupled scalar field and a perfect fluid with the equation of state of radiation ---see e.g. \cite{VKbook}. Global existence results for the Einstein field equations coupled to these fields have been given in \cite{MR1326039} and \cite{zbMATH06358183}, respectively. It should be noticed that conformal methods (more specifically, conformal rescaling) have also been used to study the occurrence of cosmological singularities \cite{MR1710659}.

\subsection*{Description of the result}

It is well-known that de Sitter and Minkowski space-times admit a conformal compactification ---that is, these space-times can be conformally embedded into a compact manifold with boundary. This compact manifold with boundary is often referred to as the unphysical space-time. The trace on the boundary of the closure of the timelike and lightlike geodesics in the unphysical space-time has two connected components (for the future and past endpoints of these geodesics). These components form a subset of the boundary which is timelike in the de Sitter case and light like in the Minkowski case. This particular structure at infinity is characteristic of a larger class of spacetimes introduced by Penrose in the 60s ---the so-called asymptotically simple space-times. Since de Sitter space-time and Minkowski space-time both admit such an asymptotic structure, it could be expected that small enough perturbations thereof would also lead to geodesically complete asymptotically simple solutions to the Einstein equations. The rigidity of the asymptotic structure and of the geodesic completeness for small enough perturbations would prove the stability of those solutions.

This approach to the stability of solutions to the Einstein equation was developed by Friedrich in the 80s. Although the Einstein equations are not conformally invariant, he proved that it is possible to incorporate a conformal degree of freedom within these equations. The resulting equations are known as the Conformal Einstein Field Equations (CFE). In the right gauge, the conformal Einstein field equations imply a symmetric hyperbolic system. Moreover, a global solution to the Einstein equations can be obtained as a local solution to the conformal field equations. These observations allowed him to prove the stability of de Sitter space-time and the semi-global stability of Minkowski space-time as a solution to the Einstein equations in vacuum in his pioneering work \cite{MR868737}.  For this reason, we focus on our presentation on the Vlasov equation and only recall the set-up and the necessary properties of the conformal Einstein field equations.

A critical aspect of the study of the Einstein equations as evolution equations is, in the first instance, the choice of the gauge. This choice guarantees that the Einstein equations become symmetric hyperbolic and that a Cauchy problem can be discussed. The approach adopted by Friedrich \cite{MR868737}  relies on the choice of a tetrad which satisfies a wave map. Following this approach, we use the tetrad formalism as in \cite{lindquist:1966el} or \cite{MR1466508} also for the Vlasov equation. The formulation of the transport equation on the cotangent bundle can then be entirely expressed in terms of the Cartan structure coefficients.

The conformal field equations require, in the presence of a matter field, control of the derivatives of the energy-momentum tensor in the frame direction. Hence, equations for these derivatives of the matter distribution need to be incorporated to the set of unknown functions to complete the system. To that end, one needs to calculate the commutator of the tetrad with the Liouville vector field. These commutators have been calculated in, for instance, \cite{zbMATH06302658} and lie at the core of the approach of \cite{2017arXiv170706141F}. The experience suggests working with the derivatives which naturally appear in the expression of the Liouville vector field, the horizontal lifts of the tetrad vectors, so that curvature terms appear in the expression of the commutator ---see \cite[Appendix]{MR3683922}.

Finally, this work relies on the use of local existence results for symmetric hyperbolic systems applied to the conformal equations coupled to the Vlasov equation. The fact that the \emph{massive} Vlasov equation is symmetric hyperbolic is discussed in many reviews \cite{MR1466508, Andreasson2011}. The local existence for the Einstein-Vlasov system has been addressed, by means of standard energy estimates, in \cite{AIF_1971__21_3_181_0}. This local existence result, tailored to handle a coupling with nonlinear wave equations and weighted Sobolev spaces (in particular to address the global existence problem) is extended in \cite{MR3186493}. It is somehow expected that the Einstein-Vlasov system (or the CFE-Vlasov system) should fit the framework developed by Kato \cite{kato}, as it was, for instance, for the Vlasov-Maxwell equations \cite{WOLLMAN1987103}.

When applying Kato's theorem to handle the local existence for the conformal field equations coupled to the massless Vlasov equations one has to be careful at several points. First, the unknowns functions have different variables. This could be typically handled by considering separately the two systems. Secondly, the nature of the coupling in the Einstein equations is integral, and it has to be verified whether this fits the framework for the non-linearity as required by Kato's result. These two problems are solved by running, based on the linear result by Kato, the same fixed point argument. Finally, and specifically for the \emph{massless} Vlasov equation, the massless equation is singular for particles with vanishing velocities. This has the consequence that the evolution ceases to be symmetric hyperbolic. This is difficulty is dealt with by proving, in the context of the aforementioned fixed point argument, a priori estimates on the support in momentum of the matter distribution.

\subsection*{Outline of the paper} Section \ref{sec:preliminaries} contains preliminaries and notation. The matter model is
 presented in Section \ref{sec:vlasovintro}; the conformal transformations are specifically addressed in Section
 \ref{subsect_conf_liou}, the equation for the derivatives of \(f\) are stated in Section \ref{sect_comm_em}, and the commutation
 formula with the stress-energy tensor are to be found in Section \ref{sect_comm_em}. The hyperbolic nature of the Vlasov system,
 made out of the matter distribution and its derivatives, is obtained in Section \ref{sec:symmetrichyperbolicvlasov}. In Section \ref{sect_cef}, the conformal Einstein field equations are recalled both in the frame and the spin formalism, and the hyperbolic reduction of the conformal field equation is described. In Section \ref{sec_ex_stab} we discuss the coupled symmetric hyperbolic system and adapt Kato's existence and stability theorems to it. Finally, the stability result for de Sitter and the semi-global stability result for Minkowski are respectively derived in Sections \ref{sect_de_sitter} and \ref{sect_minkowski}.


\numberwithin{equation}{section}
\section{Preliminaries and Notations} \label{sec:preliminaries}

\subsection{Index notation}

The signature convention for (Lorentzian) space-time metrics is $(+,-,-,-)$.  In this signature convention, the cosmological constant $\Lambda$ of de Sitter space-time is negative. {  Spacetimes with negative }cosmological constants will be said to be \emph{de Sitter-like}. \par
In what follows, Greek indices are used as coordinate indices. The Latin indices ${}_a ,\,{}_b ,{}_c,\,\ldots $ are used as abstract tensor indices while the boldface Latin indices ${}_\bma,\, {}_\bmb,{}_\bmc, \ldots$ are used as space-time frame indices taking the values $0, \ldots , 3$. In this way, given a frame $\{\bme_{\bma}\}$ a generic tensor is denoted by $T_{ab}$ while its components in the given frame are denoted by $T_{\bma \bmb}\equiv T_{ab}\bme_{\bma}{}^{a}\bme_{\bmb}{}^{b}$. We reserve the indices ${}_\bmi,\, {}_\bmj,\, {}_\bmk,\ldots$ to denote frame spatial indices respect to an adapted spacelike frame taking the values $1,\,2,\,3$. If an object is a tensor, we write it in bold font, i.e.~for a vector field we write $\bm X = X^\mu \bmpartial_\mu$, respectively, $ X^a= X^\bmb e_\bmb{}^a$. \par
Moreover we make systematic use of spinors. We follow the conventions and notation of Penrose and Rindler \cite{PenRin84}. In particular, ${}_A,\,{}_B,\,{}_C,\ldots$ are abstract spinorial indices while ${}_\bmA,\, {}_\bmB,\, {}_\bmC,\ldots$ will denote frame spinorial indices with respect to some specified spin dyad $\{\epsilon_\bmA{}^{A}\}$. Greek indices ${}_\mu,\,{}_\nu,\,{}_\lambda$ will denote coordinate indices with respect to some local chart. \par
Our conventions for the curvature tensors are fixed by the relation
\begin{equation}
(\nabla_a \nabla_b -\nabla_b \nabla_a) v^c = R^c{}_{dab} v^d.
\end{equation}

\subsection{The Einstein cylinder} \label{sec_einstein_cylinder}

Next we introduce the Einstein cylinder $\mathscr E = (\mathbb R \times \mathbb S^3, \bm g_{\mathscr E})$. In order to endow it with coordinates we view $\mathbb S^3$ as being embedded in $\mathbb R^4$ with coordinates $x^1, \dots, x^4$, i.e.~$\mathbb S^3 = \{(x^1,\dots, x^4) \in \mathbb R^4 \,:\, (x^1)^2 + \dots + (x^4)^2 = 1\}$. On any open neighbourhood of $\mathbb S^3$, one can then choose three of the four coordinates $x^1, \dots, x^4$ in order to obtain an atlas on $\mathbb S^3$. Completing this atlas with the time coordinate $\tau$ yields an atlas on the Einstein cylinder. Sometimes we also will make use of the local coordinate chart $\psi \in [0,\pi]$, $\theta \in [0,\pi]$, $\varphi \in [0,2\pi)$ of angular coordinates. The metric $\bm g_{\mathscr E}$ of the Einstein cylinder is given by
\begin{equation}
\bm g_{\mathscr E} = \mathbf{d}\tau \otimes \mathbf{d} \tau - \bm \sigma,
\end{equation}
where $\bm \sigma$ is the round metric on $\mathbb S^3$. Moreover, the fields
\begin{subequations}
\begin{align}
\bm c_1 &= x^1 \bmpartial_{x^4} - x^4 \bmpartial_{x^1} + x^2 \bmpartial_{x^3} - x^3 \bmpartial_{x^2}, \\
\bm c_2 &= x^1 \bmpartial_{x^3} - x^3 \bmpartial_{x^1} + x^4 \bmpartial_{x^2} - x^2 \bmpartial_{x^4}, \\
\bm c_3 &= x^1 \bmpartial_{x^2} - x^2 \bmpartial_{x^1} + x^3 \bmpartial_{x^4} - x^4 \bmpartial_{x^3}
\end{align}
\end{subequations}
in $T\mathbb R^4$ form an orthonormal frame on $\mathbb S^3$. With the field
\begin{equation}
\bm c_0 = \bmpartial_\tau
\end{equation}
the collection $\{\bm c_1, \bm c_2, \bm c_3\}$ can be completed to an orthonormal frame on the Einstein cylinder. \par
Both de Sitter space-time, denoted by $(\tilde{\mathcal M}_{\mathrm{dS}}, \tilde{\bm g}_\mathrm{dS})$, and Minkowski space-time, denoted by $(\tilde{\mathcal M}_\mathrm{M}, \tilde{\bm g}_{\mathrm M})$, can be identified with the interior of certain compact submanifolds of the Einstein cylinder. These submanifolds are denoted by $\mathcal M_{\mathrm{dS}}$ and $\mathcal M_{\mathrm{M}}$. On $\mathcal M_{\mathrm{dS}}$ and $\mathcal M_{\mathrm{M}}$ we then have the relations\begin{subequations}
\begin{align}
\bm g_{\mathscr E} &= \Xi_{\mathrm{dS}}^2 \, \tilde{\bm g}_\mathrm{dS}, \qquad \Xi_{\mathrm{dS}} = \cos(\tau), \\
\bm g_{\mathscr E} &= \Xi_{\mathrm{M}}^2 \, \tilde{\bm g}_{\mathrm M}, \qquad \Xi_{\mathrm{M}} = \cos(\tau) + \cos\psi,
\end{align}
\end{subequations}
respectively, where the conformal factors $\Xi_{\mathrm{dS}}$ and $\Xi_{\mathrm{M}}$ are functions globally defined on $\mathscr E$.  {More details are provided in}  \ref{sect_de_sitter} and \ref{sect_minkowski}, below. The submanifolds $\mathcal M_{\mathrm{dS}}$ and $\mathcal M_{\mathrm{M}}$ are the maximal connected submanifolds of $\mathscr E$ where $\Xi_{\mathrm{dS}} \geq 0$ or $\Xi_{\mathrm{M}} \geq 0$, respectively. \par

As already outlined in the introduction, the strategy in this article, as in \cite{MR868737} for the vacuum case, is to consider the Cauchy problem for the conformal Einstein field equations with massless Vlasov matter with initial data provided on $\mathbb S^3$. We will prove that the solution to this Cauchy problem provides a manifold $ \mathcal M $ endowed with a metric $\bm g$ and a conformal factor $\Xi$ defined on $\mathcal M$. This manifold $\mathcal M$, which we call the {\em unphysical manifold}, is diffeomorphic to $[\tau_0, \tau_\bullet] \times \mathbb S^3$, where $\tau_0 \in [-\pi/2,\tau_\bullet)$ and $\tau_\bullet \in [3\pi/4, \infty)$. The region on $\tilde{\mathcal M}$ where $\Xi > 0$ we call the {\em physical manifold} and the metric $\tilde{\bm g}$ given by the relation
\begin{equation}
\tilde{\bm g} = \Xi^{-2}\bm g
\end{equation}
we call the {\em physical metric}. The space-time $(\tilde{\mathcal M}, \tilde{\bm g})$ will then be a perturbed version of de Sitter space-time or Minkowski space-time, respectively. 

\subsection{Orthonormal frames}

We denote by $\{\bm e_\bma \}_{\bma = 0}^3$ and $\{\tilde{\bm e}_\bma\}_{\bma = 0}^3$ the orthonormal frames on $T\tilde{\mathcal M}$ which correspond to $\bm g$ and $\tilde{\bm g}$, respectively. The corresponding co-frames we denote by $\{\bmalpha^\bma\}_{\bma=0}^3$ and $\{\tilde\bmalpha^\bma\}_{\bma = 0}^3$, respectively. The co-frames $\{\bmalpha^\bma\}_{\bma=0}^3$ and $\{\tilde\bmalpha^\bma\}_{\bma = 0}^3$ induce coordinates on $T^* \tilde{\mathcal M}$ in the following way. Let $x \in \tilde{\mathcal M}$ and $\bmp \in T^*_x \tilde{\mathcal M}$. Then we write
\begin{equation} \label{def_coordinates}
\bmp = p_\mu \mathbf{d}x^\mu = v_\bma \bmalpha^\bma = \tilde v_\bma \tilde \bmalpha^\bma.
\end{equation}
The relation between those coordinates is
\begin{equation} \label{rel_coordinates}
\Xi e_\bma{}^a = \tilde e_\bma{}^a, \qquad \alpha^\bma{}_a = \Xi \tilde \alpha^\bma{}_a,  \qquad \Xi v_\bma = \tilde v_\bma.
\end{equation}

\section{Massless Vlasov matter in the cotangent bundle formulation} \label{sec:vlasovintro}

\subsection{Introduction of the model} \label{sect_intro_vlasov}

On the cotangent bundle $T^*\tilde {\mathcal  M}$ of the (physical) manifold $\tilde{\mathcal M}$ we define the particle distribution function $f \in C^1(T^*\tilde{\mathcal M}; \mathbb R)$. Integrated over a volume in phase space it gives the number of particles in the corresponding volume of physical space which have momentum in the corresponding range in momentum space. \par
The Vlasov matter model describes an ensemble of freely falling particles. In other words, the particles are assumed to move through phase-space $T\tilde{\mathcal  M}$ along the integral curves of the Liouville vector field $\tilde{\mathcal L}$. This behaviour is captured by the Vlasov equation
\begin{equation} \label{vlasov_eq}
\tilde{\mathcal L} f = 0.
\end{equation}
On $T^*\tilde{\mathcal M}$ we have the canonical coordinates $(x^\mu, p_\nu)$, $\mu, \nu = 0,\dots, 3$ ---i.e.~a co-vector $\bmv\in T_x^* \tilde{\mathcal  M}$, $x\in \tilde{\mathcal  M}$, has the form $\bmv = p_\mu \mathbf{d} x^\mu |_{x}$. In these coordinates the Liouville vector field $\tilde{\mathcal L}$ reads
\begin{equation} \label{def_liouville}
\tilde{\mathcal L} = \tilde g^{\mu\nu} p_\mu \bmpartial_{x^\nu} - \frac 12 \partial_{x^\gamma} \tilde g^{\alpha\beta} \, p_\alpha p_\beta \, \bmpartial_{p_\gamma}.
\end{equation}
The quantity $m \geq 0$, given by
\begin{equation} \label{mass_sehll_rel}
m^2 \equiv -\tilde g^{\mu\nu}(x) p_\mu p_\nu, \qquad x \in \tilde{\mathcal M}, \quad \bmp\in T_x^* \tilde{\mathcal M},
\end{equation}
is interpreted as the rest mass of the particles. It can be shown that it stays constant along the characteristic curves of the Vlasov equation (\ref{vlasov_eq}). For this reason the particle distribution function $f$ can be assumed to be supported on the co-mass shell $\tilde{\mathcal  P}_m$, a seven dimensional submanifold of $T^*\tilde{\mathcal M}$, which defined to be
\begin{equation}
\tilde{\mathcal  P}_m \equiv \{(x,\bmp)\in T^*\tilde{\mathcal M}\;:\; \tilde{\bm g}_x^{-1}(\bmp,\bmp) = -m^2, \; \bmp\,\mathrm{is\,future\,pointing}\}.
\end{equation}
If $f$ is supported on $\tilde{\mathcal  P}_m$ then it describes the distribution of particles which all have the same rest mass $m$. In the remainder of this article we assume $m=0$ and we denote the corresponding mass shell simply by $\tilde{\mathcal P}$. The particle distribution function $f$ gives rise to an energy momentum tensor $\tilde{\bm T}$ via
\begin{equation} \label{def_em_tensor}
\tilde{\bm T}[f](x) = \tilde T_{\alpha\beta}[f](x) \,  \mathbf{d}x^\alpha \mathbf{d}x^\beta, \quad x \in \tilde{\mathcal M},
\end{equation}
where
\begin{equation} \label{def_em_tensor_c}
T_{\alpha\beta}[f](x) = 8\pi \int_{\tilde{\mathcal  P}_x} f(x,p)\, p_\alpha p_\beta \, \tilde{\mathrm{d}\mathrm{vol}}_{\tilde{\mathcal  P}_x},
\end{equation}
where $\tilde{\mathrm{d}\mathrm{vol}}_{\tilde{\mathcal P}_x}$ is the volume form on the mass shell fibre $\tilde{\mathcal  P}_x$ over $x\in \tilde{\mathcal  M}$. It can be expressed by
\begin{equation} \label{vol_el_coords}
\tilde{\mathrm{d}\mathrm{vol}}_{\tilde{\mathcal P}_x} = \frac{\sqrt{|\det(g^{\mu\nu})|}}{g^{0\mu}p_\mu}  \mathrm{d}p_1  \mathrm{d}p_2  \mathrm{d}p_3.
\end{equation}
We see from the mass shell relation (\ref{mass_sehll_rel}) that the massless Vlasov matter model gives rise to a trace free energy momentum tensor. \par
In the remainder of this section we discuss the Vlasov matter model for massless particles in more detail and we show that it is an amenable matter model for being integrated into the hyperbolic reduction procedure described in Section \ref{sect_cef}, below. 
\begin{remark} {\em Since we work with massless particles, the matter distribution is a mapping over the massless mass shell $\mathcal{P}$. Nonetheless, to ease the calculations, we consider a matter distribution defined on the cotangent bundle, and then restrict it to the mass shell. It is then necessary, when differentiating derivatives, to take derivatives parallel to the mass shell. This is the case of the Liouville vector field $\mathcal{L}$, as well as other derivatives afterwards, in particular, the horizontal derivatives \eqref{def_lift_frame}. To check that these vectors are tangent to the mass shell, it is sufficient to either check they lie in the kernel of the differential of the Hamiltonian ${(x, \bm p)\mapsto g^{-1}_x(\bm p,\bm p)}$, or orthogonal to $\mathcal{L}$ for the symplectic product (or orthogonal to the null vector $v_a \partial_{v_a}$ for the Sasaki metric), using the symmetry relation \eqref{eq:symmegamma}.}
\end{remark}

\subsection{Conformal properties of the Liouville vector field} \label{subsect_conf_liou}

In Section \ref{sect_intro_vlasov} the Liouville vector field and the energy momentum tensor have been introduced. Both of these objects are metric-dependent ---the metric shows up directly in the definition (\ref{def_liouville}) of the Liouville vector field and the volume form in the definition (\ref{def_em_tensor_c}) of the energy-momentum tensor. \par
If, on a manifold $\tilde{\mathcal M}$, one has two metric tensors $\tilde{\bm g}$ and $\bm g$ which are conformally related via
\begin{equation} \label{conf_rel}
\bm g = \Xi^2 \tilde{\bm g} 
\end{equation}
each of these metrics gives rise to a different Liouville vector field and energy-momentum tensor. We give now the relation between those quantities and prove that if a particle distribution function $f \in C^1(T^* \tilde{\mathcal M}, \mathbb R_+)$ solves the Vlasov equation for one metric in the conformal class, it solves the Vlasov equations for all metrics in the conformal class. \par
To this end we consider on $T^* \tilde{\mathcal M}$  the orthonormal frames $\{\tilde{\bm e}_\bma{}\}_{\bma = 0}^3$ and $\{\bm e_\bma{}\}_{\bma = 0}^3$ corresponding to $\tilde{\bm  g}$ and $\bm g$, respectively. The corresponding co-frames are denoted by $\{\tilde{\bm \alpha}^\bmb{}\}_{\bmb = 0}^3$ and $\{\bm \alpha^\bmb{}\}_{\bmb = 0}^3$, respectively. Both frames give rise to coordinates $(x^\mu, \tilde v_\bma)$ and $(x^\mu, v_\bma)$ on $T^* \tilde{\mathcal M}$, respectively, cf.~(\ref{def_coordinates}) and (\ref{rel_coordinates}). Note, further, that with respect to this frame the mass shell relation reads
\begin{align}
&(x^\mu, v_\bma) \in \mathcal P \quad \Leftrightarrow \quad v_0 = -|v|, \qquad |v| \equiv \sqrt{(v_1)^2 + (v_2)^2 + (v_3)^2}, \\
&(x^\mu, \tilde v_\bma) \in \tilde{\mathcal P} \quad \Leftrightarrow \quad \tilde v_0 = -\left| \tilde v \right|, \qquad \left| \tilde v \right| \equiv \sqrt{(\tilde v_1)^2 + (\tilde v_2)^2 + (\tilde v_3)^2}.
\end{align}
We observe that
\begin{equation}
\tilde{\mathcal P} = \mathcal P
\end{equation}
as a manifold. The definition of the Liouville vector field depends on the metric. We have already defined $\tilde {\mathcal L}$ for $\tilde{\bm g}$ in equation (\ref{def_liouville}). For each of the metrics $\tilde{\bm g}$ and $\bm g$ we consider the corresponding Liouville vector field $\tilde{\mathcal L}$ and $\mathcal L$, given by
\begin{subequations}

\begin{align}
\tilde{\mathcal L} &\equiv \tilde g^{\mu\nu} p_\mu \bmpartial_{x^\nu} - \frac 12 \partial_{x^\gamma} \tilde g^{\alpha\beta} \, p_\alpha p_\beta \, \bmpartial_{p_\gamma}, \\
\mathcal L &\equiv  g^{\mu\nu} p_\mu \bmpartial_{x^\nu} - \frac 12 \partial_{x^\gamma} g^{\alpha\beta} \, p_\alpha p_\beta \, \bmpartial_{p_\gamma}.
\end{align}

\end{subequations}
\begin{lemma}
We have\begin{subequations}
\begin{align}
\tilde {\mathcal L} &= \eta^{\bma \bmb} \tilde v_\bma \tilde e_\bmb{}^\mu \bmpartial_{x^\mu} + \eta^{\bma\bmb} \tilde v_\bma \tilde v_\bmc \, \tilde \Gamma_\bmb{}^\bmc{}_\bmd \, \bmpartial_{\tilde v_\bmd}, \quad \tilde \Gamma_\bma{}^\bmc{}_\bmb \equiv \tilde \bmalpha^\bmc\left(\tilde \nabla_{\tilde{\bm e}_\bma} \tilde{\bm e}_\bmb \right), \label{for_vlasov_1} \\
\mathcal L &= \eta^{\bma \bmb} v_\bma e_\bmb{}^\mu \bmpartial_{x^\mu} + \eta^{\bma\bmb} v_\bma v_\bmc \, \Gamma_\bmb{}^\bmc{}_\bmd \, \bmpartial_{v_\bmd}, \quad \Gamma_\bma{}^\bmc{}_\bmb \equiv \bmalpha^\bmc \left(\nabla_{\bm e_\bma} \bm e_\bmb \right). \label{for_vlasov_2}
\end{align}
\end{subequations}
\end{lemma}

\begin{proof}
The result is only proved for $\mathcal L$ since the calculation goes completely analogously for $\tilde{\mathcal L}$. We consider the coordinate transformation
\begin{equation}
x^\mu \mapsto y^\mu(x) = x^\mu, \quad p_\alpha \mapsto v_\bma(x,p) = e_\bma{}^\alpha p_\alpha.
\end{equation}
This implies
\begin{equation}
\bmpartial_{p_\alpha} = e_\bma{}^\alpha \bmpartial_{v_\bma}, \qquad \partial_{x^\mu} = \bmpartial_{y^\mu} + (\partial_\mu e_\bma{}^\alpha ) \alpha^\bmb{}_\alpha v_\bmb \, \bmpartial_{v_\bma}.
\end{equation}
Then, using $g^{\alpha\beta} = \eta^{\bma\bmb} e_\bma{}^\alpha e_\bmb^\beta$ and $p_\mu = \alpha^\bma{}_\mu v_\bma$ we calculate
\begin{align}
g^{\mu\nu} p_\mu \bmpartial_{x^\nu} &= \eta^{\bma\bmb} e_\bmb{}^\nu v_\bma \bmpartial_{y^\nu} + \eta^{\bma\bmb} v_\bma e_\bmb{}^\nu \left(\partial_\nu e_ \bmc{}^\beta\right) \alpha^\bmd{}_\beta v_ \bmd \bmpartial_{v_\bmc}, \\
-\frac 12 \partial_{x^\gamma} g^{\alpha\beta} p_\alpha p_\beta \bmpartial_{p_\gamma} &= -\eta^{\bma\bmb} v_\bma e_\bmc{}^\gamma \left(\partial_\gamma e_\bmb{}^\beta \right) \alpha^\bmd{}_\beta v_\bmd \bmpartial_{v_\bmc}.
\end{align}
Combined this yields
\begin{equation}
\mathcal L = \eta^{\bma\bmb} e_\bmb{}^\nu v_\bma \bmpartial_{y^\nu} + \eta^{\bma\bmb} v_\bma \left(e_\bmb{}^\nu\left(\partial_\nu e_\bmc{}^\beta\right) - e_\bmc{}^\gamma \left(\partial_\gamma e_\bmb{}^\beta\right) \right) \alpha^\bmd{}_\beta v_\bmd \bmpartial_{v_\bmc}.
\end{equation}
Now, noting that for the bracket we have
\begin{equation}
\left(e_\bmb{}^\nu\left(\partial_\nu e_\bmc{}^\beta\right) - e_\bmc{}^\gamma\left(\partial_\gamma e_\bmb{}^\beta\right) \right) \alpha^\bmd{}_\beta = [\bm e_\bmb, \bm e_\bmc]^\bmd.
\end{equation}
Using Cartan's structure equation, $[\bm e_\bmb, \bm e_\bmc]^\bmd = \Gamma_{[\bmb}{}^\bmd{}_{\bmc]}$, and the symmetry
\begin{equation}\label{eq:symmegamma}
\Gamma_\bma{}^\bmd{}_\bmb \eta_{\bmd\bmc} + \Gamma_\bma{}^\bmd{}_\bmc \eta_{\bmd\bmb} = 0
\end{equation}
which holds for any connection which is metric (cf.~\cite{VKbook}), we arrive at the asserted formula for $\mathcal L$ (after renaming the coordinates $y^\mu$ as $x^\mu$).
\end{proof}

\begin{lemma}
On $\mathcal{P}$ we have
\begin{equation}
\tilde {\mathcal L} = \Xi^2 \mathcal L.
\end{equation}
\end{lemma}

\begin{proof}
Using the relations (\ref{rel_coordinates}) the frames and co-frames with respect to $\bm g$ and $\tilde{\bm g}$, respectively, we derive
\begin{align}
\tilde \Gamma_\bmb{}^\bmc{}_\bmd &= \tilde \bmalpha^\bmc \left( \tilde\nabla_{\tilde{\bm e}_\bmb} \tilde{\bm e}_{\bmd} \right) = \frac{1}{\Xi} \bmalpha^\bmc \left(\Xi \tilde\nabla_{\bm e_\bmb} (\Xi \bm e_\bmd)\right) \\
&= \bmalpha^\bmc \left(\bm e_\bmb(\Xi) \bm e_\bmd\right) + \Xi \bmalpha^\bmc \left(\tilde \nabla_{\bm e_\bmb}\right) \\
&= \bm e_\bmb(\Xi) \delta_\bmd^\bmc + \Xi \left(\Gamma_\bmb{}^\bmc{}_\bmd + Q_\bmb{}^\bmc{}_\bmd \right),
\end{align}
where $Q_\bmb{}^\bmc{}_\bmd$ is the transition tensor, in the case of our conformal transformation given by
\begin{equation}
Q_\bmb{}^\bmc{}_\bmd = \frac{1}{\Xi} \left(\bm e_\bmb(\Xi)\delta_\bmd^\bmc + \bm e_\bmd(\Xi) \delta_\bmb^\bmc - \bm e_\bma(\Xi) \eta^{\bma\bmc} \eta_{\bmb\bmd} \right),
\end{equation}
cf.~\cite{VKbook}. This yields
\begin{equation} \label{for_tilde_christoffel}
\tilde \Gamma_\bmb{}^\bmc{}_\bmd = \Xi \Gamma_\bmb{}^\bmc{}_\bmd + 2 \bm e_\bmb(\Xi)\delta_\bmd^\bmc + \bm e_\bmd(\Xi) \delta_\bmb^\bmc - \bm e_\bma(\Xi) \eta^{\bma\bmc} \eta_{\bmb\bmd}.
\end{equation}
We consider the Liouville vector field $\tilde{\mathcal L}$ and we want to perform the change of coordinates given by
\begin{equation}
x^\mu \mapsto y^\mu(x) = x^\mu, \quad \tilde v_\bma \mapsto v_\bma(x, \tilde v_\bma) = \frac{1}{\Xi} \tilde v_\bma,
\end{equation}
cf.~the relations (\ref{rel_coordinates}). This yields
\begin{equation} \label{for_tilde_ccords}
\bmpartial_{\tilde v_\bma} = \frac{1}{\Xi} \bmpartial_{v_\bma}, \qquad \bmpartial_{x^\mu} = \bmpartial_{y^\mu} + \Xi v_\bma \left(\partial_\mu \frac{1}{\Xi} \right) \bmpartial_{v_\bma}.
\end{equation}
Then, using the mass shell relation $\eta^{\bma\bmb} v_\bma v_\bmb = 0$, and inserting (\ref{for_tilde_christoffel}) and (\ref{for_tilde_ccords}) into the formula (\ref{for_vlasov_1}) for $\mathcal L$, we obtain the following assertion:
\end{proof}

\begin{corollary}
For any function $f \in C^1(\mathcal{P}; \mathbb R_+)$ we have
\begin{equation}
\tilde {\mathcal L} f = 0 \quad \Leftrightarrow \quad \mathcal L f = 0.
\end{equation}
\end{corollary}

Let $f\in C^1(T^* \tilde{\mathcal M}; \mathbb R_+)$. For different coordinate maps the particle distribution is realised by different functions $\mathbb R^7 \to \mathbb R$. Let $\bmp \in T^* \tilde{\mathcal M}$ be an arbitrary point with coordinates $(x^\mu, \tilde v_\bma)$ and $(x^\mu, v_\bma)$, respectively. We denote (abusing notation)
\begin{equation} \label{eq_coordinates}
f(\bmp) = f(x^\mu, p_\nu) = \tilde{\mathfrak f}(x^\mu, \tilde v_\bma) = \mathfrak f(x^\mu, v_\bma).
\end{equation}
The relation (\ref{rel_coordinates}) yields
\begin{equation} \label{scaling_law_f}
\mathfrak f(x^\mu, v_\bma) = \tilde{\mathfrak f} \left( x^\mu, \Xi v_\bma \right).
\end{equation}

\begin{remark}
{\em The scaling law (\ref{scaling_law_f}) has a concrete physical meaning. We construct here conformal compactifications of de Sitter like or Minkowski like space-times. These space-times are compactly embedded into the Einstein cylinder on which we have the (almost) flat unphysical metric $\bm g$ which is related to $\tilde{\bm g}$ via (\ref{conf_rel}) and the conformal factor $\Xi$ approaches zero as one approaches the boundary of this embedding. The massless Vlasov particles move along null geodesics which, as sets, are invariant under the conformal rescaling. The parametrisation of the geodesics, however, depends on the metric. As it will be shown later, measured with respect to the unphysical metric, the (absolute value of the) velocities only change a little. This means, by (\ref{scaling_law_f}), that the velocity with respect to the physical metric approaches zero. This is consistent since the conformal boundary, representing null infinity, cannot be reached in finite time.}
\end{remark}

Next, we consider the energy-momentum tensor. Recall the definition on the physical manifold
\begin{equation}
\tilde T_{ab}[f](x) \equiv 8\pi \int_{\tilde{\mathcal P}_x} f(x,\bmp)\, p_a p_b \, \tilde{\mathrm{dvol}}_{\tilde {\mathcal P}_x}, \quad (x,\bmp) \in T^* \tilde{\mathcal M}.
\end{equation}
The domain of integration is unchanged by the conformal rescaling (\ref{conf_rel}) since $\mathcal P = \tilde{\mathcal P}$. However, the volume form depends on the metric. So we define as well
\begin{equation}
T_{ab}[f](x) \equiv 8\pi \int_{\mathcal P_x} f(x,\bmp)\, p_a p_b \, \tilde{\mathrm{dvol}}_{\mathcal P_x}, \quad (x,\bmp) \in T^* \mathcal M.
\end{equation}
A calculation shows
\begin{align}
\tilde T_{\bma\bmb}[f](x) &= \tilde T_{ab}[f](x) \tilde e_\bma{}^a \tilde e_\bmb{}^b  = -8 \pi \int_{\mathrm R_v^3} \tilde{\mathfrak f}(x^\mu, \tilde v_\bma) \frac{\tilde v_\bma \tilde v_\bmb}{|\tilde v|} \, \mathrm d\tilde v_1 \mathrm d\tilde v_2 \mathrm d\tilde v_3, \\
T_{\bma\bmb}[f](x) &= T_{ab}[f](x) e_\bma{}^a e_\bmb{}^b = -8 \pi \int_{\mathrm R_v^3} \mathfrak f(x^\mu, v_\bma) \frac{v_\bma v_\bmb}{|v|} \, \mathrm dv_1 \mathrm dv_2 \mathrm dv_3.
\end{align}
Note that with respect to frame coordinates $g^{0\mu} p_\mu = -v_0 = |v|$ and therefore volume form is given by
\begin{equation}
\mathrm{dvol}_{\mathcal P_x} = \frac{1}{|v|} \mathrm dv_1 \mathrm dv_2 \mathrm dv_3.
\end{equation}

\begin{lemma} \label{lem_em_vlasov}
Let $f\in C^1(\mathcal{P}; \mathbb R_+)$. We have
\begin{equation}
\tilde{\bm T}[f] = \Xi^2 \bm T[f].
\end{equation}
\end{lemma}

\begin{proof}
Write $\tilde{\bm T} = \tilde T_{\bma\bmb} \tilde{\bm \alpha}^\bma \tilde{\bm \alpha}^\bmb$ and $\bm T = T_{\bma\bmb} \bm \alpha^\bma \bm \alpha^\bmb$. Performing a change of variables $\tilde v_\bma \mapsto \Xi v_\bma$ and using the scaling law (\ref{scaling_law_f}) we obtain $\tilde T_{\bma\bmb} = \Xi^4 T_{\bma\bmb}$. Since $\tilde{\bm \alpha}^\bma = \Xi^{-1} \bm \alpha^\bma$ we obtain the asserted relation.
\end{proof}

\subsection{Commuting with the Liouville vector field}

In the remainder of this section we only consider the unphysical quantities. We define
\begin{equation} \label{def_lift_frame}
\hat{\bm e}_\bma \equiv \bm e_\bma +  v_\bmc \Gamma_\bma{}^\bmc{}_\bmd \, \bmpartial_{v_\bmd},
\end{equation}
as the horizontal lifts of the frame fields $\bm e_\bma$ to the cotangent bundle $T^*\tilde{\mathcal M}$. For details on the horizontal lifts, see \cite{lindquist:1966el}. Observe that these vectors are tangent to the mass shell. Note that then the formula (\ref{for_vlasov_2}) for the Liouville vector field with respect to the orthonormal frame $\bm e_a$ reads
\begin{equation}
\mathcal L = \eta^{\bma\bmb} v_\bma \hat{\bm e}_\bmb.
\end{equation}

\begin{lemma}
The Liouville vector field satisfies the commutator formulas
\begin{subequations}
\begin{align}
[\mathcal L, \hat{\bm e}_\bmg] &= \eta^{\bma\bmb} \left(v_\bma \Gamma_{[\bmb}{}^\bmc{}_{\bmg]} \hat{\bm e}_\bmc - v_\bmc \Gamma_{\bmg}{}^\bmc{}_\bma \hat{\bm e}_\bmb \right)  \label{com_for_hat}\\
&\quad + \eta^{\bma\bmb} v_\bma v_\bmc \left(\Xi d^\bmc{}_{\bmd\bmb\bmg} + 2 S_{\bmd[\bmb}{}^{\bmc\bme} L_{\bmg]\bme} \right) \bmpartial_{v_\bmd}, \nonumber \\
[\mathcal L, \partial_{v_\bmg}] &= - \eta^{\bmg\bma} \hat{\bm e}_\bma - \eta^{\bma\bmb} v_\bma \Gamma_\bmb{}^\bmg{}_\bmc \, \bmpartial_{v_\bmc}. \label{com_for_v}
\end{align}
\end{subequations}
\end{lemma}

\begin{proof}
First we verify formula (\ref{com_for_hat}). We calculate
\begin{align*}
[\mathcal L, \hat{\bm e}_\bmg] &= \eta^{\bma\bmb} v_\bma [\bm e_\bmb, \bm e_\bmg] + \eta^{\bma\bmb} \left[v_\bma \bm e_\bmb, v_\bme \Gamma_\bmg{}^\bme{}_\bmf \bmpartial_{v_\bmf} \right] - \eta^{\bma\bmb} v_\bma v_\bmc \bm e_\bmg \left( \Gamma_\bmb{}^\bmc{}_\bmd \right) \bmpartial_{v_\bmd} \\
&\quad + \eta^{\bma\bmb} \Gamma_\bmb{}^\bmc{}_\bmd \Gamma_\bmg{}^\bme{}_\bmf \left[ v_\bma v_\bme \bmpartial_{v_\bmd}, v_\bme \bmpartial_{v_\bmf} \right] \nonumber \\
&= \eta^{\bma\bmb} v_\bma [\bm e_\bmb, \bm e_\bmg] + \eta^{\bma\bmb} v_\bma v_\bme \, \bm e_{[\bmb} \left( \Gamma_{\bmg]}{}^\bme{}_\bmf \right) \bmpartial_{v_\bmf}\\ 
&\quad - \eta^{\bma\bmb} v_\bme \Gamma_\bmg{}^\bme{}_\bma \bm e_\bmb + \eta^{\bma\bmb} \Gamma_\bmb{}^\bmc{}_\bme \Gamma_\bmg{}^\bme{}_\bmf v_\bma v_\bmc \bmpartial_{v_\bmf} \\
&\quad - \eta^{\bma\bmb} v_\bme v_\bmc \Gamma_\bmb{}^\bmc{}_\bmd \Gamma_\bmg{}^\bme{}_\bma \bmpartial_{v_\bmd} - \eta^{\bma\bmb} v_\bme v_\bma \Gamma_\bmb{}^\bmc{}_\bmd \Gamma_\bmg{}^\bme{}_\bmc \bmpartial_{v_\bmd} \nonumber
\end{align*}
There are two terms which are not given in terms of the variables $x^\mu, v_\bma$ or the geometric fields $\bm e_\bma$, $\Gamma_\bma{}^\bmc{}_\bmb$. Using curvature identities we can however make them disappear. These are the no torsion condition and the expression of the Riemann tensor in terms of the frame field and the Christoffel symbols, namely
\begin{align}
0 &= \Sigma_\bma{}^\bmc{}_\bmb = \Gamma_\bma{}^\bmc{}_\bmb- \Gamma_\bmb{}^\bmc{}_\bma - \alpha^\bmc \left([\bm e_\bma, \bm e_\bmb] \right), \\
R^\bmc{}_{\bmd\bma\bmb} &= \bm e_{[\bma}\left( \Gamma_{\bmb]}{}^\bmc{}_\bmd\right) + \Gamma_\bmf{}^\bmc{}_\bmd \Gamma_{[\bmb}{}^\bmf{}_{\bma]} +  \Gamma_\bmb{}^\bmf{}_\bmd \Gamma_\bma{}^\bmc{}_\bmf - \Gamma_\bma{}^\bmf{}_\bmd \Gamma_\bmb{}^\bmc{}_\bmf,
\end{align}
cf.~\cite{VKbook}. Simplification and decomposing the Riemann tensor as
\begin{equation}
R^\bmc{}_{\bmd\bma\bmb}= \Xi d^\bmc{}_{\bmd\bma\bmb} + 2 S_{\bmd[\bma}{}^{\bmc\bme} L_{\bmb]\bme}
\end{equation}
yields formula (\ref{com_for_hat}). Equation (\ref{com_for_v}) can be obtained by a straightforward calculation.
\end{proof}

\subsection{Commuting with the energy momentum tensor} \label{sect_comm_em}

As it will be described in the sections below, the source terms of the conformal Einstein field equations involve terms depending on the matter fields. These terms are the components $T_{\bma\bmb}$, $\nabla_\bma T_{\bmb\bmc}$ of the energy momentum tensor and its covariant derivative, respectively. It is important to express these components without derivatives of the unknowns. For this reason we consider the components of the covariant derivatives in more detail. Using the definition (\ref{def_lift_frame}) of the lifts of the frame vector fields we obtain the formula
\begin{align}
\nabla_\bma T_{\bmb\bmc}[f] &= T_{\bmb\bmc}[\bme_\bma(f)] - \Gamma_\bma{}^\bmd{}_\bmb T_{\bmd\bmc}[f] - \Gamma_\bma{}^\bmd{}_\bmc T_{\bmb\bmd}[f]\\
&= T_{\bmb\bmc}[\hat{\bm e}_\bma(f)] + 8\pi \Gamma_\bma{}^\bme{}_\bmd \int_{\mathbb R_v^3} (\partial_{v_\bmd}  \mathfrak f(x^\mu, v_\bma) ) \, \frac{v_\bme v_\bmb v_\bmc}{|v|}\, \mathrm dv_1 \mathrm dv_2 \mathrm dv_3 \\
&\quad - \Gamma_\bma{}^\bmd{}_\bmb T_{\bmd\bmc}[f] - \Gamma_\bma{}^\bmd{}_\bmc T_{\bmb\bmd}[f]. \nonumber
\end{align}
This motivates to include the fields $\hat{\bm e}_\bma(f)$ and $\partial_{v_\bma}f$ to the collection of unknowns.

\section{The symmetric hyperbolic system for massless Vlasov matter} \label{sec:symmetrichyperbolicvlasov}

\subsection{System for the matter fields}
The conformal Einstein field equations (CEF) with massless Vlasov matter will be formulated and solved in the coordinates $(\tau, \underline x, v^\bma)$, where $\underline x$ are coordinates on $\mathbb S^3$ and $v^\bma$, $\bma = 0,\dots, 3$ are the coordinates induced by the $\bm g$-frame $\{\bm e_\bma\}_{\bma = 0}^3$. \par
The matter field will be described by the unknowns
\begin{equation} \label{def_uf}
\mathbf{u}_{\mathfrak f} = (\mathfrak f, \mathfrak f_\bma, \varphi^\bma),
\end{equation}
where
\begin{equation} \label{def_matter_quan}
 \mathfrak f_\bma \equiv \bm e_\bma{}^\alpha \partial_{x^\alpha} \mathfrak f + v_\bmc \Gamma_\bma{}^\bmc{}_\bmd \partial_{v_\bmd} \mathfrak f, \quad \varphi^\bma \equiv \partial_{v_\bma} \mathfrak f.
\end{equation}
In the remainder of this article these functions will be referred to as {\em the matter fields}. Denote by $N_{\mathfrak f}$ the number of components of $\bm u_{\mathfrak f}$ ---there are nine components.

\begin{remark} \label{rem:constraintsmatter} It is important to note that the initial for $\bm u_{\mathfrak f}$ are not independent. Indeed, the initial data for the field $\varphi^\bma$ can be calculated from the derivatives of $\mathfrak{f}$ initially. Hence, it is necessary to impose, initially, constraints on the initial data of this system. It is then necessary to check that these constraints are propagated. The propagation of this constraints is done in Section \ref{sec:constraints}
\end{remark}

Now, we consider the equation for the matter fields $\mathfrak f$, $\mathfrak f_\bma$, $\varphi^\bma$. In frame coordinates these equations read
\begin{subequations}
\begin{align}
&\eta^{\bma \bmb} v_\bma e_\bmb{}^\mu \bmpartial_{x^\mu} \mathfrak f + \eta^{\bma\bmb} v_\bma v_\bmc \, \Gamma_\bmb{}^\bmc{}_\bmd \, \bmpartial_{v_\bmd} \mathfrak f = 0, \label{conf_f_1}\\
&\eta^{\bma \bmb} v_\bma e_\bmb{}^\mu \bmpartial_{x^\mu} \mathfrak f_\bmg + \eta^{\bma\bmb} v_\bma v_\bmc \, \Gamma_\bmb{}^\bmc{}_\bmd \, \bmpartial_{v_\bmd} \mathfrak f_\bmg \label{conf_f_2} \\
&\quad = \eta^{\bma\bmb} \left(v_\bma \Gamma_{[\bmb}{}^\bmc{}_{\bmg]} \mathfrak f_\bmc - v_\bmc \Gamma_{\bmg}{}^\bmc{}_\bma \mathfrak f_\bmb \right) \nonumber \\
&\quad\quad  + \eta^{\bma\bmb} v_\bma v_\bmc \left(\Xi d^\bmc{}_{\bmd\bmb\bmg} + 2 S_{\bmd[\bmb}{}^{\bmc\bme} L_{\bmg]\bme} \right) \varphi^\bmd, \nonumber \\
&\eta^{\bma \bmb} v_\bma e_\bmb{}^\mu \bmpartial_{x^\mu} \varphi^\bmg + \eta^{\bma\bmb} v_\bma v_\bmc \, \Gamma_\bmb{}^\bmc{}_\bmd \, \bmpartial_{v_\bmd} \varphi^\bmg = - \eta^{\bmg\bma} \mathfrak f_\bma - \eta^{\bma\bmb} v_\bma \Gamma_\bmb{}^\bmg{}_\bmc \varphi^\bmc. \label{conf_f_3}
\end{align}
\end{subequations}
The first equation is just the Vlasov equation in the frame coordinates, cf.~(\ref{for_vlasov_2}). The last two equations follow directly from the commutator formulas (\ref{com_for_hat}) and (\ref{com_for_v}). \par
We want to write these equations in a more compact form. To this end, we denote the collection of all geometric fields, i.e.~the metric and concomitants, by $\mathbf{u}_g$ ---see the definition in (\ref{def_ug}), below. We introduce the matrices $\bm A_{\mathfrak f}^\mu[\mathbf{u}_g(\tau, \cdot)](\underline x, v_\bma)$, $\mathfrak A_\bmb[\mathbf{u}_g(\tau, \cdot)](\underline x, v_\bma)$, $\mu, \bmb = 0, \dots, 3$, and the vector $\bm F_{\mathfrak f}[\mathbf{u}_g(\tau, \cdot), \mathbf{u}_{\mathfrak f}(\tau, \cdot)](\underline x, v_\bma)$ such that the equations (\ref{conf_f_1})--(\ref{conf_f_2}) can be written as
\begin{multline} \label{matter_system}
\bm A_{\mathfrak f}^0[\mathbf{u}_g(\tau, \cdot)]\partial_{x^0} \mathbf{u}_{\mathfrak f} + \bm A_{\mathfrak f}^i[\mathbf{u}_g(\tau, \cdot)] \partial_{x^i} \mathbf{u}_{\mathfrak f} + \mathfrak A_\bmc[\mathbf{u}(\tau, \cdot)] \partial_{v_\bmc} \mathbf{u}_{\mathfrak f} \\ = \bm F_{\mathfrak f}[\mathbf{u}_g(\tau, \cdot), \mathbf{u}_{\mathfrak f}(\tau, \cdot)],
\end{multline}
where
\begin{equation} \label{def_f_source}
\begin{split}
& \bm F_{\mathfrak f}[\mathbf{u}_g(\tau, \cdot), \mathbf{u}_{\mathfrak f}(\tau, \cdot)]= \\&\left(\begin{array}{c} 0 \\
\eta^{\bma\bmb} \left(v_\bma \Gamma_{[\bmb}{}^\bmc{}_{\bm 0]} \mathfrak f_\bmc - v_\bmc \Gamma_{\bm 0}{}^\bmc{}_\bma \mathfrak f_\bmb \right)+ \eta^{\bma\bmb} v_\bma v_\bmc \left(\Xi d^\bmc{}_{\bmd\bmb\bm 0} + 2 S_{\bmd[\bmb}{}^{\bmc\bme} L_{\bm 0]\bme} \right) \varphi^\bmd \\
\vdots \\
\eta^{\bma\bmb} \left(v_\bma \Gamma_{[\bmb}{}^\bmc{}_{\bm 3]} \mathfrak f_\bmc - v_\bmc \Gamma_{\bm 3}{}^\bmc{}_\bma \mathfrak f_\bmb \right)+ \eta^{\bma\bmb} v_\bma v_\bmc \left(\Xi d^\bmc{}_{\bmd\bmb\bm 3} + 2 S_{\bmd[\bmb}{}^{\bmc\bme} L_{\bm 3]\bme} \right) \varphi^\bmd \\
- \eta^{\bm 0\bma} \mathfrak f_\bma - \eta^{\bma\bmb} v_\bma \Gamma_\bmb{}^{\bm 0}{}_\bmc \varphi^\bmc \\
\vdots \\
- \eta^{\bm 3\bma} \mathfrak f_\bma - \eta^{\bma\bmb} v_\bma \Gamma_\bmb{}^{\bm 3}{}_\bmc \varphi^\bmc
\end{array}\right)
\end{split}
\end{equation}
and
\begin{equation} \label{def_a_matrices}
\begin{aligned}
\bm A_{\mathfrak f}^\mu[\mathbf{u}_g(\tau, \cdot)](\underline x, v_\bma) &= \eta^{\bma\bmb} v_\bma e_\bmb{}^\mu \bm I_7, \\
\mathfrak A_\bmd[\mathbf{u}_g(\tau, \cdot)](\underline x, v_\bma) &= \eta^{\bma\bmb} v_\bma v_\bmc \Gamma_\bmb{}^\bmc{}_\bmd \bm I_7,
\end{aligned}
\end{equation}
where $\bm I_7$ is the unit matrix in $\mathbb R^7$.

\begin{remark}
{\em Observe that the matrix $\bm A_{\mathfrak f}^0[\mathbf{u}_g(\tau, \cdot)]$ looses rank if $v_\bma=0$. Hence, whenever $v_\bma=0$ the system \eqref{matter_system} is not symmetric hyperbolic.}
\end{remark}

\section{The conformal Einstein field equations and their hyperbolic reduction} \label{sect_cef}

In this section, we state the conformal Einstein field equations with matter which are satisfied by the unphysical metric $\bm g$. The system of equations will first be stated in general coordinates and then with respect to an orthonormal $\bm g$-frame and finally in the spinor formalism. This will be the starting point for the hyperbolic reduction procedure described in the next section. \par
The statements in this section are valid for any matter model giving rise to a trace-free energy momentum tensor, so in particular for massless Vlasov matter.

\subsection{The metric equations}

In the following let $(\tilde{\mathcal{M}},\tilde{g}_{ab})$ denote a space-time satisfying the Einstein field equations
\begin{equation} \label{e_f_eq}
\tilde{R}_{ab} - \frac{1}{2}\tilde{R}\tilde{g}_{ab} + \lambda \tilde{g}_{ab} = \tilde{T}_{ab}
\end{equation}
with trace-free matter, that is,
\begin{equation}
\tilde{T}_a{}^a =0.
\end{equation}
Moreover, let $(\mathcal{M},g_{ab})$ denote a conformally related (unphysical) space-time such that
\begin{equation} 
g_{ab}= \Xi^2 \tilde{g}_{ab}.
\end{equation}
An \emph{unphysical energy-momentum} $T_{ab}$ is defined through the relation
\begin{equation}
T_{ab} \equiv \Xi^{-2} \tilde{T}_{ab}.
\end{equation}
This definition is consistent with massless Vlasov matter, cf.~Lemma \ref{lem_em_vlasov}. It can be readily verified that
\begin{equation}
\nabla^a T_{ab}=0.
\end{equation}

For reference, we list here the metric conformal Einstein field equations coupled to trace-free matter:
\begin{eqnarray*}
&& \nabla_a \nabla_b \Xi = -\Xi L_{ab} + s g_{ab} +\tfrac{1}{2}\Xi^3T_{ab}, \\
&&  \nabla_a s = -L_{ac} \nabla^c \Xi + \tfrac{1}{2} \Xi^2 \nabla^c \Xi
 T_{ac}  + \tfrac{1}{6} \Xi^3 \nabla^c T_{ca}, \\
&& \nabla_c L_{db} - \nabla_d L_{cb} = \nabla_a \Xi d^a{}_{bcd} + \Xi
T_{cdb},  \\
&& \nabla_a d^a{}_{bcd} = T_{cdb}, \\
&& 6 \Xi s - 3 \nabla_c \Xi \nabla^c \Xi =\lambda.
\end{eqnarray*}
In the previous equations
\begin{equation}
T_{abc} \equiv \Xi \nabla_{[a} T_{b]c} + \nabla_{[a}\Xi T_{b]c} +
g_{c[a} T_{b]d}\nabla^d\Xi,
\end{equation}
denotes the \emph{rescaled Cotton tensor}. {The unknowns of the metric conformal Einstein field equations are  the \emph{Schouten tensor} $L_{ab}$, the \emph{Friedrich scalar} $s$, and the \emph{rescaled Weyl tensor}  $d^a{}_{bcd}$. In terms of the conformal factor $\Xi$ and the Riemann tensor $R^c{}_{dab}$, these concomitants of the metric $g$ can be expressed in the following way. The Schouten tensor is given by
\begin{equation}
L_{ab} = \frac 12 R_{ab} - \frac 12 R g_{ab},
\end{equation}
where $R_{db}=R^a{}_{dab}$ is the Ricci tensor and $R=g^{ab}R_{ab}$ is the Ricci scalar. The Weyl tensor $C^c{}_{dab}$ is defined by the relation
\begin{equation}
R^c{}_{dab} = C^c{}_{dab} + 2S_{d[a}{}^{ce} L_{b]e},
\end{equation}
where
$$
S_{ab}{}^{cd} = \delta_a{}^c \delta_b{}^d - g_{ab} g^{cd},
$$
and the rescaled Weyl tensor $d^c{}_{dab}$ is given by
\begin{equation}
C^c{}_{dab} = \Xi d^c{}_{dab}. 
\end{equation}
The Friedrich scalar is defined by
\begin{equation}
s=\frac 14 \nabla^c \nabla_c \Xi + \frac{1}{24} R\Xi.
\end{equation}
See \cite{VKbook} for details. }

\subsection{The spinorial formulation of the equations}
In this section, we use Penrose's spinorial notation and conventions. The hyperbolic reduction of the CFE is more conveniently discussed in terms of their spinorial formulation.  The reader can refer to \cite{VKbook} for further details. Since the machinery is heavy, and only the end result is of relevance for this work, we just state with no further details the equations. The main ideas behind this hyperbolic reduction procedure can be traced back to  the work by Friedrich \cite{MR1131434}.

Given an orthonormal frame $\{ \bme_\bma{}\}_{\bma = 0}^3$ one constructs a Newman-Penrose (NP) frame, that is to say a null frame, $\{ \bme_{\bmA\bmA'} \}_{\bmA = 0,1, \bmA' =0', 1'}$. { The Infeld-van der Waerden symbols $\sigma_{\bm a}{}^{\bm A \bm A'}$, $\bm a = 0,\dots, 3$, $\bm A = 0,1$, $\bm A' = 0',1'$ mediate between these frames, i.e.~
\begin{equation}
\bm e_{\bm A\bm A'} = \sigma^{\bm a}{}_{\bm A \bm A'} \bm e_{\bm a}.
\end{equation}
} The spinorial components of the fields
\begin{equation}
\Sigma_a{}^c{}_b, \qquad R^c{}_{dab},\qquad
T_{ab}, \qquad L_{ab}, \qquad d^a{}_{bcd},
\qquad T_{abc}, \qquad \Gamma_{a}{}^b{}_c,
\end{equation}
will be denoted, respectively, by
\begin{eqnarray*}
& \Sigma_{\bmA\bmA'}{}^{\bmC\bmC'}{}_{\bmB\bmB'}, \quad
R^{\bmC\bmC'}{}_{\bmD\bmD'\bmA\bmA'\bmB\bmB'}, \quad
T_{\bmA\bmA'\bmB\bmB'},& \\
& \quad L_{\bmA\bmA'\bmB\bmB'},\quad
d^{\bmA\bmA'}{}_{\bmB\bmB'\bmC\bmC'\bmD\bmD'}, \quad
  T_{\bmA\bmA'\bmB\bmB'\bmC\bmC'}, \quad \Gamma_{\bmA\bmA'}{}^{\bmB\bmB'}{}_{\bmC\bmC'}.&
\end{eqnarray*}
{The relation between tensor components, and the spinorial components is then obtained through the Infeld-van der Waerden symbols,} {i.e.~we have
$$
\Sigma_{\bmA\bmA'}{}^{\bmC\bmC'}{}_{\bmB\bmB'} = \sigma^{\bm a}{}_{\bm A \bm A'} \sigma_{\bm c}{}^{\bm C \bm C'} \sigma^{\bm b}{}_{\bm B \bm B'} \Sigma_a{}^c{}_b
$$
etc. See \cite[Section 3.1.9]{VKbook} for a definition of the Infeld-van der Waerden symbols through a spinor basis. They can be explicitly written in terms of the Pauli matrices as
\begin{align*}
\left(\sigma_{\bm 0}^{\bm A \bm A'}\right) &\equiv \frac{1}{\sqrt 2} \left(\begin{array}{cc}1 & 0\\ 0 & 1\end{array}\right), & \left(\sigma_{\bm 1}^{\bm A \bm A'}\right) &\equiv \frac{1}{\sqrt 2} \left(\begin{array}{cc}0 & 1\\ 1 & 0\end{array}\right), \\
\left(\sigma_{\bm 2}^{\bm A \bm A'}\right) &\equiv \frac{1}{\sqrt 2} \left(\begin{array}{cc}0 & i\\ -i & 0\end{array}\right), & \left(\sigma_{\bm 3}^{\bm A \bm A'}\right) &\equiv \frac{1}{\sqrt 2} \left(\begin{array}{cc}1 & 0\\ 0 & -1\end{array}\right),
\end{align*}
and
\begin{align*}
\left(\sigma^{\bm 0}_{\bm A \bm A'}\right) &\equiv \frac{1}{\sqrt 2} \left(\begin{array}{cc}1 & 0\\ 0 & 1\end{array}\right), & \left(\sigma^{\bm 1}_{\bm A \bm A'}\right) &\equiv \frac{1}{\sqrt 2} \left(\begin{array}{cc}0 & 1\\ 1 & 0\end{array}\right), \\
\left(\sigma^{\bm 2}_{\bm A \bm A'}\right) &\equiv \frac{1}{\sqrt 2} \left(\begin{array}{cc}0 & -i\\ i & 0\end{array}\right), & \left(\sigma^{\bm 3}_{\bm A \bm A'}\right) &\equiv \frac{1}{\sqrt 2} \left(\begin{array}{cc}1 & 0\\ 0 & -1\end{array}\right).
\end{align*}
Furthermore}, the connection coefficients can be decomposed in terms of reduced spin connection coefficients as
\begin{equation}
\Gamma_{\bmA\bmA'}{}^{\bmB\bmB'}{}_{\bmC\bmC'} =
\Gamma_{\bmA\bmA'}{}^\bmB{}_\bmC \delta_{\bmC'}{}^{\bmB'}+ \bar{\Gamma}_{\bmA\bmA'}{}^{\bmB'}{}_{\bmC'}\delta_\bmC{}^\bmB,
\end{equation}
where $\Gamma_{AA'BC}=\Gamma_{AA'(BC)}$ as the connection is metric. The curvature spinor can be written as
\begin{equation}
R^{\bmC\bmC'}{}_{\bmD\bmD'\bmA\bmA'\bmB\bmB'} =
R^\bmC{}_{\bmD\bmA\bmA'\bmB\bmB'} \delta_{\bmD'}{}^{\bmC'} +
\bar{R}^{\bmC'}{}_{\bmD'\bmA\bmA'\bmB\bmB'} \delta_\bmD{}^\bmC,
\end{equation}
where $R_{\bmC\bmD\bmA\bmA'\bmB\bmB'} =R_{(\bmC\bmD)\bmA\bmA'\bmB\bmB'}$. Their expression in terms of the reduced spin connection coefficients is given by
\begin{align*}
R^\bmC{}_{\bmD\bmA\bmA'\bmB\bmB'} &=
\bme_{\bmA\bmA'}(\Gamma_{\bmB\bmB'}{}^\bmC{}_\bmD)
-\bme_{\bmB\bmB'}(\Gamma_{\bmA\bmA'}{}^\bmC{}_\bmD) \\
&\quad -\Gamma_{\bmF\bmB'}{}^\bmC{}_\bmD \Gamma_{\bmA\bmA'}{}^\bmF{}_\bmB -
\Gamma_{\bmB\bmF'}{}^\bmC{}_\bmD
\bar{\Gamma}_{\bmA\bmA'}{}^{\bmF'}{}_{\bmB'} \\
&\quad + \Gamma_{\bmF\bmA'}{}^\bmC{}_\bmD \Gamma_{\bmB\bmB'}{}^\bmF{}_\bmA +\Gamma_{\bmA\bmF'}{}^\bmC{}_\bmD \bar{\Gamma}_{\bmB\bmB'}{}^{\bmF'}{}_{\bmA'} \\
&\quad + \Gamma_{\bmA\bmA'}{}^\bmC{}_{\bmF} \Gamma_{\bmB\bmB'}{}^\bmF{}_\bmD -\Gamma_{\bmB\bmB'}{}^\bmC{}_\bmF \Gamma_{\bmA\bmA'}{}^\bmF{}_\bmD.
\end{align*}
Moreover, in view of its symmetries one can write
\begin{equation}
R^\bmC{}_{\bmD\bmA\bmA'\bmB\bmB'} = -\Xi \phi^\bmC{}_{\bmD\bmA\bmB}
\epsilon_{\bmA'\bmB'} +L_{\bmD\bmB'\bmA\bmA'}\delta_\bmB{}^\bmC - L_{\bmD\bmA'\bmB\bmB'}\delta_\bmA{}^\bmC,
\end{equation}
with $\phi_{\bmA\bmB\bmC\bmD}$ and $L_{\bmA\bmA'\bmB\bmB'}$ the components of the Weyl and Schouten spinors, respectively.

The spinorial counterpart of $T_{\bma\bmb}$ satisfies, in view of its vanishing trace, the property
\[
T_{\bmA\bmA'\bmB\bmB'}=
T_{(\bmA\bmB)(\bmA'\bmB')}.
\]
Finally, exploiting the anti-symmetry $T_{cdb}=-T_{dcb}$ of the rescaled Cotton tensor, one has the split
\[
T_{\bmC\bmC'\bmD\bmD'\bmB\bmB'} = T_{\bmC\bmD\bmB\bmB'}
\epsilon_{\bmC'\bmD'} + \bar{T}_{\bmC'\bmD'\bmB\bmB'} \epsilon_{\bmC\bmD},
\]
where $T_{\bmC\bmD\bmB\bmB'}\equiv \tfrac{1}{2}
T_{\bmC\bmQ'\bmD}{}^{\bmQ'}{}_{\bmB\bmB'}$. Observe that $T_{\bmC\bmD\bmB\bmB'}=T_{(\bmC\bmD)\bmB\bmB'}$. \par
The spinorial counterparts of the frame conformal Einstein field equations are obtained by suitable contractions with the Infeld-van der Waerden symbols. Simpler expressions are obtained if one takes into account the remarks made in the previous subsection. The conformal field equations in the spinorial formalism can then be stated, in an arbitrary Newman-Penrose tetrad \((\bme_{\bmA\bmA'}) \):
\begin{subequations}
\begin{align}
   [\bme_{\bmA\bmA'},\bme_{\bmB\bmB'}]
   -(\Gamma_{\bmA\bmA'}{}^{\bmC\bmC'}{}_{\bmB\bmB'}-\Gamma_{\bmB\bmB'}{}^{\bmC\bmC'}{}_{\bmA\bmA'})\bme_{\bmC\bmC'} &= 0 \label{SpinorialZQ1} \\
   \bme_{\bmA\bmA'}(\Gamma_{\bmB\bmB'}{}^\bmC{}_\bmD)
   -\bme_{\bmB\bmB'}(\Gamma_{\bmA\bmA'}{}^\bmC{}_\bmD) &&\nonumber\\
   -\Gamma_{\bmF\bmB'}{}^\bmC{}_\bmD \Gamma_{\bmA\bmA'}{}^\bmF{}_\bmB
   -\Gamma_{\bmB\bmF'}{}^\bmC{}_\bmD
   \bar{\Gamma}_{\bmA\bmA'}{}^{\bmF'}{}_{\bmB'} + \Gamma_{\bmF\bmA'}{}^\bmC{}_\bmD \Gamma_{\bmB\bmB'}{}^\bmF{}_\bmA&&\nonumber\\
   +\Gamma_{\bmA\bmF'}{}^\bmC{}_\bmD
   \bar{\Gamma}_{\bmB\bmB'}{}^{\bmF'}{}_{\bmA'}
   +\Gamma_{\bmA\bmA'}{}^\bmC{}_{\bmF}
   \Gamma_{\bmB\bmB'}{}^\bmF{}_\bmD -\Gamma_{\bmB\bmB'}{}^\bmC{}_\bmF
   \Gamma_{\bmA\bmA'}{}^\bmF{}_\bmD&&\nonumber\\
   +\Xi
   \phi^\bmC{}_{\bmD\bmA\bmB}\epsilon_{\bmA'\bmB'}
   -L_{\bmD\bmB'\bmA\bmA'}\delta_\bmB{}^\bmC +
   L_{\bmD\bmA'\bmB\bmB'}\delta_\bmA{}^\bmC &= 0, \label{SpinorialZQ2}\\
  \nabla_{\bmA\bmA'} \nabla_{\bmB\bmB'}\Xi  +\Xi L_{\bmA\bmA'\bmB\bmB'} - s
 \epsilon_{\bmA\bmB}\epsilon_{\bmA'\bmB'}
   -\tfrac{1}{2}\Xi^3T_{\bmA\bmA'\bmB\bmB'} &= 0, \label{SpinorialZQ3}\\
 \nabla_{\bmA\bmA'} s + L_{\bmA\bmA'\bmC\bmC'}
   \nabla^{\bmC\bmC'} \Xi  && \nonumber\\- \tfrac{1}{2} \Xi^2
   \nabla^{\bmC\bmC'} \Xi T_{\bmA\bmA'\bmC\bmC'}  - \tfrac{1}{6} \Xi^3
   \nabla^{\bmC\bmC'} T_{\bmA\bmA'\bmC\bmC'}&= 0, \label{SpinorialZQ4}\\
\nabla_{\bmC\bmC'}
   L_{\bmD\bmD'\bmB\bmB'} - \nabla_{\bmD\bmD'} L_{\bmC\bmC'\bmB\bmB'}
   -\nabla_{\bmA\bmA'}\Xi
   d^{\bmA\bmA'}{}_{\bmB\bmB'\bmC\bmC'\bmD\bmD'}\nonumber\\
 \hspace{3cm}- \Xi
T_{\bmC\bmC'\bmD\bmD'\bmB\bmB'} &= 0, \label{SpinorialZQ5} \\
 \nabla_{\bmA\bmA'}
   d^{\bmA\bmA'}{}_{\bmB\bmB'\bmC\bmC'\bmD\bmD'} -
   T_{\bmC\bmC'\bmD\bmD'\bmB\bmB'} &= 0, \\
 6 \Xi s - 3 \nabla_{\bmC\bmC'}\Xi  \nabla^{\bmC\bmC'} \Xi - \lambda &= 0\label{SpinorialZQ6}.
\end{align}
\end{subequations}
The connection between the conformal Einstein field equations and the Einstein field equations is given by the following:

\begin{proposition} \label{Proposition:FrameCFEImplyEFE}
Let
\[
(\bme_{\bmA\bmA'}, \Gamma_{\bmA\bmA'}{}^\bmB{}_\bmC, \Xi,
s,L_{\bmA\bmA'\bmB\bmB'}, \phi_{\bmA\bmB\bmC\bmD},T_{\bmA\bmA'\bmB\bmB'})
\]
denote a solution to the frame conformal field equations with
\[
\nabla^{\bmA\bmA'} T_{\bmA\bmA'\bmB\bmB'}=0
\]
and such that, on an open set $\mathcal{U}\subset\mathcal{M}$,
\[
\Xi\neq0, \quad \det
(\epsilon^{\bmA\bmB}\epsilon^{\bmA'\bmB'} \bme_{\bmA\bmA'}{} \otimes \bme_{\bmB\bmB'}) \neq 0,
\]
where $\epsilon$ is the symplectic product on spinors related to the unphysical metric by $\bmg _{\bma\bmb} = \epsilon_{\bmA
\bmB}\epsilon_{\bmA'\bmB'}$. Then the metric
\[
\tilde{\bmg} =
\Xi^{-2} \epsilon_{\bmA\bmB}\epsilon_{\bmA'\bmB'} \bmomega^{\bmA\bmA'}\otimes \bmomega^{\bmA'\bmB'}
\]
where $\{\bmomega^{\bmA\bmA'}\}_{\bma = 0}^3$ is the dual frame to $\{\bme_{\bmA\bmA'}\}_{\bma = 0}^3$, is a solution to the Einstein field equations (\ref{e_f_eq}) on $\mathcal{U}$.
\end{proposition}

\subsection{Basic set-up for the frame} \label{sec_basic_setup}

In the following, all the calculations will be performed in an open subset $\mathcal{U}\subset \mathcal{M}$ of $(\mathcal{M},\bmg)$. On $\mathcal{U}$ one considers some local coordinates $x=(x^\mu)$ and an arbitrary frame $\{ \bmc_\bma \}_{\bma = 0}^3$ which may or may not be a coordinate frame. Let $\{ \bmalpha^\bma \}_{\bma = 0}^3$ denote the dual co-frame so that $\langle \bmalpha^\bma, \bmc_\bmb\rangle =\delta_\bmb{}^\bma$. Moreover, let $\bmnabla$ denote the Levi-Civita covariant derivative of the metric $\bmg$.

It will be assumed that $\mathcal{U}$ is covered by a non-singular congruence of curves with tangent vector $\bmtau$ satisfying the normalisation condition $\bmg(\bmtau,\bmtau)=2$. The vector $\bmtau$ does not need to be hypersurface orthogonal. Let $\tau^{AA'}$ denote the spinorial counterpart of $\tau^a$. We restrict attention to spin bases $\{ \epsilon_\bmA{}^A \}$ satisfying the condition
\[
\tau^{\bmA\bmA'} = \epsilon_\bmzero{}^{\bmA}
\epsilon_{\bmzero'}{}^{\bmA'} + \epsilon_\bmone{}^\bmA
\epsilon_{\bmone'}{}^{\bmA'}.
\]
All spinors will be expressed in components with respect to this class of spin bases. \par
Let $\{ \bme_{\bmA\bmA'}\}$ and $\{  \bmomega^{\bmA\bmA'} \}$  denote, respectively, the null frame and co-frame associated to the spin basis $\{\epsilon_\bmA{}^A \}$. One therefore has that
\[
\langle \bmomega^{\bmA\bmA'}, \bme_{\bmB\bmB'} \rangle =
\epsilon_\bmB{}^\bmA \epsilon_{\bmB'}{}^{\bmA'}.
\]
 At every point $p\in \mathcal{U}$ a basis of the subspace of $T|_p(\mathcal{U})$ orthogonal to $\bmtau$ is given by $\bme_{\bmA\bmB} = \tau_{(\bmB}{}^{\bmA'} \bme_{\bmA)\bmA'}$. The \emph{spatial} frame can be expanded in terms of the vectors $\bmc_\bma$ as $\bme_{\bmA\bmB} = e_{\bmA\bmB}{}^\bma \bmc_\bma$. Moreover, one also has that
\[
\bme_{\bmA\bmA'} = e_{\bmA\bmA'}{}^\bma
\bmc_\bma.
\]
The frame coefficients can be decomposed using $\tau^{\bmA\bmA'}$ as
\[
\bme_{\bmA\bmA'}{}^\bma = \tfrac{1}{2}\tau_{\bmA\bmA'}
-\tau^\bmQ{}_{\bmA'} e_{\bmA\bmQ}{}^\bma,
\]
with
\[
e^\bma \equiv \tau^{\bmA\bmA'} e_{\bmA\bmA'}{}^\bma, \qquad
e_{\bmA\bmB}{}^\bma\equiv \tau_{(\bmA}{}^{\bmA'}e_{\bmB)\bmA'}{}^\bma.
\]

\subsection{Gauge source functions}\label{sec:gauge}

Following the hyperbolic reduction procedure introduced in \cite[Section 6]{MR1131434}, see also \cite[Section 13.2.2]{VKbook}, we consider \emph{gauge source functions} $F^\bma(x)$, $F_{\bmA\bmB}(x)$ and $R(x)$ such that
\begin{subequations}
\begin{align}
 \nabla^{\bmA\bmA'} \nabla_{\bmA\bmA'} e_{\bmA\bmA'}{}^\bma &= F^\bma(x), \label{gauge_f_1} \\
 \nabla^{\bmC\bmC'} \Gamma_{\bmC\bmC'\bmA\bmB} &= F_{\bmA\bmB}(x), \label{gauge_f_2} \\
\nabla^{\bmA\bmA'} L_{\bmA\bmA'\bmB\bmB'} &= \tfrac{1}{6}\nabla_{\bmB\bmB'} R(x). \label{gauge_f_3}
\end{align}
\end{subequations}
The fields $F^\bma(x)$, $F_{\bmA\bmB}(x)$ and $R(x)$ are, respectively, the \emph{coordinate gauge source function}, the \emph{frame gauge source} and the \emph{conformal gauge source function}. In particular, $R(x)$ can be identified with the Ricci scalar of the unphysical metric $\bmg$.

\subsection{Symmetric hyperbolic form of the CEF}\label{sec:symhypcfe}

It is convenient to introduce the unknown function
\begin{equation} \label{def_ug}
\mathbf{u}_g = \left(\Xi, \; \Sigma_{\bmA\bmA'}, \; s, \; , e^\mu_{\bmA\bmA'},\;\Gamma_{\bmA\bmA'\bmB\bmC}, \; \Phi_{\bmA\bmA'\bmB\bmB'}, \; \phi_{\bmA\bmB\bmC\bmD} \right)
\end{equation}
denoting the collection of all geometric fields and $\bm u_{\mathfrak f}$ denotes the collection of all matter fields, cf.~Equation~(\ref{def_uf}). Denote the number of independent components in $\mathbf{u}_g$ by $N_g$. We can now state the reduced hyperbolic form of the Einstein equations, the theorem initially proved by Friedrich (see \cite[Section 2]{MR868737} for the first version,  \cite[Sections 3 and 6]{MR1131434} for the spinorial version, and the monograph \cite[Proposition 13.1]{VKbook})
\begin{proposition}\label{prop_cef_sh}
Given arbitrary smooth gauge source functions
\[
F^\bma(x), \qquad F_{\bmA\bmB}(x), \qquad R(x),
\]
such that
\begin{eqnarray*}
& \nabla^{\bmQ\bmQ'}e_{\bmQ\bmQ'}{}^\bma = F^\bma(x), \qquad
\nabla^{\bmQ\bmQ'} \Gamma_{\bmQ\bmQ'\bmA\bmB} = F_{\bmA\bmB}(x), & \\
&\nabla^{\bmQ\bmQ'} L_{\bmQ\bmQ'\bmB\bmB'} =
\displaystyle\frac{1}{6}\nabla_{\bmB\bmB'} R(x), &
\end{eqnarray*}
and assuming that the components of the matter tensors $T_{\bma\bmb}$ and $T_{\bma\bmb\bmc}$ can be written in such a way that they do not contain derivatives of the matter fields $\bm u_{\mathfrak f}$, then the conformal Einstein field equations \eqref{SpinorialZQ1}--\eqref{SpinorialZQ6}
imply a symmetric hyperbolic system of equations for the independent components of the geometric fields $\mathbf{u}_g$ of the form
\begin{equation} \label{cef_metric_first}
\bm A_g^0[\mathbf{u}_g(\tau, \cdot)](\underline x) \cdot \partial_\tau \mathbf{u}_g + \bm A_g^i[\mathbf{u}_g(\tau, \cdot)](\underline x) \cdot \partial_{x^i} \mathbf{u}_g= \bm F_g[\mathbf{u}_g(\tau, \cdot), \mathbf{u}_{\mathfrak f}(\tau, \cdot)](\underline x), 
\end{equation}
where 
\begin{subequations}
\begin{align}
&\bm A_g^\mu : H_{loc}^m\left(M;\, \mathbb R^{N_g}\right) \to H_{loc}^m\left(M;\, \mathbb R^{N_g \times N_g}\right), \quad \mu = 0,\dots, 3,\\
&\bm F_g: H_{loc}^m \left(M; \, \mathbb R^{N_g} \right) \times H_{loc}^m \left(\mathcal{P}; \, \mathbb R^{N_{\mathfrak f}} \right) \to H_{loc}^m\left(M; \, \mathbb R^{N_g}\right)
\end{align}
\end{subequations}
where $\mathbf{u}_{\mathfrak f} : M \to \mathbb R^{N_{\mathfrak f}}$ is mapping containing as components  the distribution matter $\mathfrak f$ and some of its derivatives, see Equation \eqref{def_uf}. The operator  \(A_g^0\) can be decomposed as
\[
\bm A_g^0 = \mathbf{I} +  \tilde{\bm A}_g^0 
\]
 where $\mathbf{I}$ is the identity of the corresponding dimension, and \(\tilde{\bm A}_g^0\) is containing the actual dependency on the metric quantity. Moreover the matrices $\bm A_g^\mu[\mathbf{z}]$ are polynomials in $\mathbf{z} \in \mathbb R^{N_g}$ of degree at most one with constant coefficients and they are symmetric.

\end{proposition}

\section{Subsidiary equations and propagation of the constraints}\label{sec:constraints}

The conformal field equations come with constraints imposed by the form of the system (see for instance \cite[Chapter 13.3]{VKbook}). This constraints relate quantities within the system of conformal field equations, and form a system of compatibility equations which need to be satisfied. In the presence of matter fields, the coupling imposes further constraints. These have two origins: the constraints coming from the presence of matter in the conformal field equations, and those directly related to the matter fields. Altogether, to verify that solutions for the system \eqref{SpinorialZQ1}-\eqref{SpinorialZQ6} coupled to the system \eqref{conf_f_1} to \eqref{conf_f_3}, it is necessary to check that the following so-called zero quantities,
\begin{subequations}
\begin{align}
\mathcal{Z}_{\bma\bmb}  & \equiv \nabla_\bma \nabla_\bmb \Xi +\Xi L_{\bma\bmb} - s g_{\bma\bmb} -  \frac12 \Xi^3 T_{\bma\bmb}, \\
\mathcal{Z}_{\bma} & \equiv \nabla_\bma s + L_{\bma\bmc}\nabla^\bmc \Xi - \frac12 \Xi^2 \nabla^\bmc \Xi T_{\bma\bmc} - \frac16 \Xi^3 \nabla^\bmc T_{\bma\bmc}, \\
\Delta_{\bmc\bmd\bmb} &\equiv \nabla_\bmc L_{\bmd\bmb} - \nabla_\bmd L_{\bmc\bmb}  - \nabla_\bma \Xi d^{\bma}{}_{\bmc\bmd\bmb} - \Xi T{}_{\bmc\bmd\bmb}, \\
\Lambda_{\bma \bmc\bmd} &\equiv \nabla_\bma  d^{\bma}{}_{\bmb\bmc\bmd} -T_{\bmc\bmd\bmb}, \\
\mathcal{Z} &\equiv 6 \Xi s -3 \nabla_\bmc \Xi \nabla^\bmc \Xi -\lambda, \\
 \Sigma_\bma{}^\bmc{}_\bmb &\equiv \Gamma_\bma{}^\bmc{}_\bmb- \Gamma_\bmb{}^\bmc{}_\bma - \alpha^\bmc \left([\bm e_\bma, \bm e_\bmb] \right), \\
\Xi^{\bmc}{}_{\bmd \bma  \bmb} &\equiv \Xi d^\bmc{}_{\bmd\bma\bmb} + 2 S_{\bmd[\bma}{}^{\bmc\bme} L_{\bmb]\bme} \\
& \quad - \left(\bm e_{[\bma}\left( \Gamma_{\bmb]}{}^\bmc{}_\bmd\right) + \Gamma_\bmf{}^\bmc{}_\bmd \Gamma_{[\bmb}{}^\bmf{}_{\bma]} +  \Gamma_\bmb{}^\bmf{}_\bmd \Gamma_\bma{}^\bmc{}_\bmf - \Gamma_\bma{}^\bmf{}_\bmd \Gamma_\bmb{}^\bmc{}_\bmf\right), \nonumber
\end{align}
\end{subequations}
completed with the matter zero quantities 
\begin{subequations}
\begin{align}
\Phi^\bma &\equiv \varphi ^\bma - \partial_{v_\bma} \mathfrak f, \\ 
F_\bma &\equiv \mathfrak f_\bma - \hat{\bm e}_\bma(\mathfrak f),
\end{align}
\end{subequations}
remain zero throughout the time evolution.

It is known that the constraints related to the Einstein equations, when the stress-energy tensor is stress-free can be recast as a symmetric hyperbolic system (see \cite{MR1131434}; also \cite[Proposition 13.2]{VKbook}). Hence, if these quantities are initially vanishing, they will remain so during the evolution. A similar calculation needs to be performed for the matter zero-quantities. We prove the following 

\begin{lemma} The zero quantities $\Phi^\bma$ and $F_{\bma}$, $\bma = 0,\dots, 3$ obey the homogeneous equations in the zero quantities:
\begin{subequations}
\begin{align}
\mathcal{L}\Phi^\bmg  &=  - \eta^{\bmg\bma} F_\bma - \eta^{\bma\bmb} v_\bma \Gamma_\bmb{}^\bmg{}_\bmc \Phi^\bmc, \label{matter_const_ev_1} \\
\mathcal{L}F_\bmg & = \eta^{\bma\bmb} \left(v_\bma \Gamma_{[\bmb}{}^\bmc{}_{\bmg]} F_\bmc - v_\bmc \Gamma_{\bmg}{}^\bmc{}_\bma F_\bmb \right) \label{matter_const_ev_2} \\
 &\quad+ \eta^{\bma\bmb} v_\bma v_\bmc \left(\Xi d^\bmc{}_{\bmd\bmb\bmg} + 2 S_{\bmd[\bmb}{}^{\bmc\bme} L_{\bmg]\bme} \right) \Phi^\bmd \nonumber \\
&\quad   + \eta^{\bma \bmb} v_{\bma} v_{\bmc} \Xi^\bmc{}_{\bmd\bmb \bmg} \left(\varphi^{\bm d} - \Phi^{\bm d}\right) \nonumber\\
&\quad+ \eta^{\bma \bmb} v_{\bma} \Sigma_{\bmb}{}^{\bmc}{}_{\bmg} \left(\mathfrak f_\bmc - F_\bmc - v_{\bmf} \Gamma_\bmc{}^\bmf{}_\bmd (\varphi^\bmd - \Phi^\bmd)\right). \nonumber
\end{align}
\end{subequations}
\end{lemma}

\begin{proof} To perform this calculation, we need to calculate the commutator between the Liouville vector field and derivatives of the matter field, without assuming that the expression of the metric is torsion-free, or that the curvature expression hold.  
We observe that, using Equation \eqref{com_for_hat},
\begin{align*}
[\mathcal{L} , \hat{\bme}_\bmg ]
&= \eta^{\bma\bmb} v_\bma \Gamma_{[\bmb}{}^\bmc{}_{\bmg]} \hat{\bm e}_\bmc - \eta^{\bma\bmb} v_\bmc \Gamma_{\bmg}{}^\bmc{}_\bma \hat{\bm e}_\bmb - \eta^{\bma\bmb} v_\bma \Sigma_{\bmb}{}^{\bmc}{}_{\bmg} \left(\hat{\bm e}_\bmc - v_\bmf \Gamma_\bmc{}^\bmf{}_\bmd \bmpartial_{v_\bmd} \right) \\
&\quad + \eta^{\bma\bmb} v_\bma v_\bmc \left(\Xi d^\bmc{}_{\bmd\bmb\bmg} + 2 S_{\bmd[\bmb}{}^{\bmc\bme} L_{\bmg]\bme}  - \Xi^\bmc{}_{\bmd\bmb \bmg}\right) \bmpartial_{v_\bmd}.
\end{align*}
We deduce that the zero quantity $F^\bma$ satisfies (\ref{matter_const_ev_2}). In a similar fashion, we observe that, using Equation \eqref{com_for_v}, we can establish (\ref{matter_const_ev_1}).
\end{proof}

\section{Existence and Stability Theory} \label{sec_ex_stab}

\subsection{The initial value problem.}

The strategy here is as outlined in \cite{MR868737}. 
The initial data which we prescribe on $\mathbb S^3$ for the Cauchy problem will be chosen to be a background space-time (de Sitter or Minkowski) plus a perturbation. The background space-times are discussed and explicitly given in Section \ref{sect_de_sitter} and \ref{sect_minkowski}, below. The conformal Einstein field equations with massless Vlasov matter are given with respect to the orthonormal $\bm g$-frame $\{\bm e_\bma\}_{\bma = 0}^3$. In order to understand the space-times provided by the Cauchy problem as perturbation of de Sitter or Minkowski we wish to have some notion of background frame on $\hat{\mathcal M}$ with whom the $\bm g$-frame $\{\bm e_\bma\}_{\bma = 0}^3$ can be compared. With the help of so called {\em cylinder maps} $\phi:\mathcal U \subset \hat{\mathcal M} \to \mathbb R \times \mathbb S^3$ (cf.~\cite{MR868737}), one can on $\hat{\mathcal M}$ define the fields $\bm c_0, \dots, \bm c_3$ in terms of the $\phi$-pull-backs of the coordinate functions $\tau, x^0, \dots x^4$ from the Einstein cylinder to $\hat{\mathcal M}$. By these means, one obtains a collection of fields which are globally defined on $\hat{\mathcal M}$. Moreover, these are orthonormal with respect to the $\phi$-pull back of the metric $\bm g_{\mathscr E}$ of the Einstein cylinder and for the $\bm g$-frame $\{\bm e_\bma\}_{\bma = 0}^3$ we have
\begin{equation}
\bm e_\bma = e_\bma{}^\bmb \bm c_\bmb.
\end{equation}
In \cite{MR868737}, it is shown that this procedure fixes the gauge source functions $F^\bma$, $F_{\bm A \bm B}$, $R$, cf.~(\ref{gauge_f_1})--(\ref{gauge_f_3}), to be
\begin{equation} \label{gsf}
F^a = 0, \quad  F_{\bm A \bm B} = 0, \quad  R = -6.
\end{equation}
We take the spinorial counterparts $e_{\bmA\bmA'}^b$ of the components $e_\bma{}^b$ of the frame fields $\bm e_\bma$ with respect to the vacuum frame fields $\bm c_\bmb$ as unknowns of the conformal Einstein field equations with massless Vlasov matter and express the other geometric fields in terms of these unknowns. Note that as this place we make a specific choice for the frame fields $\{\bmc_\bma\}_{\bma = 0}^3$ already mentioned in Section \ref{sec_basic_setup}. \par
Combining the equations (\ref{matter_system}) for the matter fields and the equations (\ref{cef_metric_first}) for the geometric fields, we obtain the the conformal Einstein field equations with massless Vlasov matter
\begin{subequations}
\begin{align}
&\bm A_g^0[\mathbf{u}_g(\tau, \cdot)](\underline x) \partial_\tau \mathbf{u}_g + \bm A_g^i[\mathbf{u}_g(\tau, \cdot)](\underline x) \partial_{x^i} \mathbf{u}_g = \bm F_g[\mathbf{u}_g(\tau, \cdot), \mathbf{u}_{\mathfrak f}(\tau, \cdot)](\underline x), \label{cefmvm1}\\
&\bm A_{\mathfrak f}^0[\mathbf{u}_g(\tau, \cdot)](\underline x, v_\bma) \partial_{\tau} \mathbf{u}_{\mathfrak f} + \bm A_{\mathfrak f}^i[\mathbf{u}_g(\tau, \cdot)](\underline x, v_\bma) \partial_{x^i} \mathbf{u}_{\mathfrak f} \label{cefmvm2}  \\
&\hspace{2cm} + \mathfrak A_\bmc[\mathbf{u}_g(\tau, \cdot)](\underline x, v_\bma) \partial_{v_\bmc} \mathbf{u}_{\mathfrak f} = \bm F_{\mathfrak f}[\mathbf{u}_f(\tau, \cdot), \mathbf{u}_{\mathfrak f}(\tau, \cdot)](\underline x, v_\bma). \nonumber
\end{align}
\end{subequations}
We recall that a solution $\mathbf{u} = (\mathbf{u}_g, \mathbf{u}_{\mathfrak f})$ of the system (\ref{cefmvm1})--(\ref{cefmvm2}) consists of the geometric fields in $\mathbf{u}_g$ as given by (\ref{def_ug}), and matter fields in $\mathbf{u}_{\mathfrak f}$, as given by (\ref{def_uf}). \par
In order to apply Kato's theorem, we have to assume that, initially, the coefficient matrices and the source terms are bounded from above, and that the $\bm A^0$-matrices are bounded away from zero. This could, a priori, be problematic if the range of the $\bm v_a$-variables is neither bounded above, nor away from the zero velocity. Let $\delta \in (0,1)$ and define
\begin{equation} \label{def_dom_omega}
\Omega_\delta \equiv \left\{ \left(v_{\bm 0}, v_{\bm 1}, v_{\bm 2}, v_{\bm 3}\right) \in \mathbb R^4 \, : \, v_{\bm 0} = -|v|, \; \delta \leq |v| \leq \frac{1}{\delta} \right\}.
\end{equation}
In order to avoid the aforementioned boundedness issues we consider the matter fields as functions from the space $H^m_0(\mathbb S^3 \times \Omega_{1/4};\, \mathbb R^{N_{\mathfrak f}})$ and we regard the coefficients and the source terms in the system (\ref{cefmvm1})--(\ref{cefmvm2}) as operators
\begin{subequations}
\begin{align}
&\bm A_g^\mu : H^m\left(\mathbb S^3;\, \mathbb R^{N_g}\right) \to H^m\left(\mathbb S^3;\, \mathbb R^{N_g \times N_g}\right), \quad \mu = 0,\dots, 3, \label{operators_1}\\
&\bm A_{\mathfrak f}^\nu, \mathfrak A_\bma : H^m\left(\mathbb S^3;\, \mathbb R^{N_g}\right)  \to H_0^m \left( \mathbb S^3 \times \Omega_{1/4}; \, \mathbb R^{N_{\mathfrak f} \times N_{\mathfrak f}}\right), \quad \nu, \bma = 0, \dots, 3, \label{operators_2}\\
&\bm F_g: H^m \left(\mathbb S^3; \, \mathbb R^{N_g} \right) \times H_0^m \left(\mathbb S^3 \times \Omega_{1 / 4}; \, \mathbb R^{N_{\mathfrak f}} \right) \to H^m\left(\mathbb S^3; \, \mathbb R^{N_g}\right), \label{operators_3} \\
&\bm F_{\mathfrak f} : H^m \left(\mathbb S^3; \, \mathbb R^{N_g} \right) \times H_0^m \left(\mathbb S^3 \times \Omega_{1 / 4}; \, \mathbb R^{N_{\mathfrak f}} \right) \to H_0^m\left(\mathbb S^3 \times \Omega_{1/4}; \, \mathbb R^{N_{\mathfrak f}}\right). \label{operators_4}
\end{align}
\end{subequations}
\begin{remark} 
{\em
One has the following:
\begin{enumerate}[(i)]
    \item We assume later that $m\geq 5$. At this level of regularity, the considered Sobolev spaces are algebras. Since all the operators considered above are algebraic expressions in the components $\mathbf{u}_g$, the integral in a compact $v$ set of the components of $\mathbf{u}_\mathfrak{f}$, and of the components of $\mathbf{u}_\mathfrak{f}$, they take values in the corresponding Sobolev space. 
    \item Kato \cite{kato} considers uniformly local Sobolev spaces, which are not relevant for us since we are working with spaces which are compact in space and velocities.
    \item Moreover we remark that the operators (\ref{operators_1})--(\ref{operators_4}) are special cases of the operators considered in \cite{kato}. There the operators are denoted by $G_j(t)$, $F(t)$ and they are assumed to be non-linear operators sending functions with values in a Hilbert space $P$ to functions with values in $\mathscr B(P)$ or $P$, respectively. These non-linear operators have an explicit $t$-dependence. The operators (\ref{operators_1})--(\ref{operators_4}) however have no explicit $t$-dependence and the Hilbert space is $\mathbb R^{N_g}$ or $\mathbb R^{N_{\mathfrak f}}$, respectively.
\end{enumerate}}
\end{remark}
We will see that for the class of initial data which we prescribe for the system (\ref{cefmvm1})--(\ref{cefmvm2}), if the $v$-support is bounded initially in $\Omega_{1/2}$, say, it will remain bounded sufficiently long in $\Omega_{1/4}$.  We now discuss this class of initial data. To this end, we use the notation
\begin{equation}
\mathring{\mathbf{u}} = ( \mathring{\mathbf{u}}_g, \mathring {\mathbf{u}}_{\mathfrak f} ) \in H^{m+1}(\mathbb S^3 \times \Omega_{1/4}; \, \mathbb R^{N_g+N_{\mathfrak f}}), \quad \mathrm{where} \; \mathring{\mathbf{u}}_{\mathfrak f} = 0
\end{equation}
to denote a {\em background solution}. For technical reasons, we need to assume higher regularity which is available since this background solution describes either the de Sitter or the Minkowski space-time (see \cite[Theorem I and Eq. (3.29)]{kato}). The metric functions $\mathring{\mathbf{u}}_g$ are explicitly given in Sections \ref{sect_de_sitter} and \ref{sect_minkowski} below. At this stage, it is only important that these functions solve the conformal Einstein field equations with massless Vlasov matter (\ref{cefmvm1})--(\ref{cefmvm2}) on the whole Einstein cylinder $\mathbb R \times \mathbb S^3$. Denote furthermore, the initial data for $\mathring{\mathbf{u}}_g$ by $\mathring{\mathbf{u}}_g^\star$ ---this data is prescribed on the whole of $\mathbb{S}^3$.  \par
Let $\varepsilon, \delta >0$ and define for $\tau \in [0,\pi]$
\begin{subequations}
\begin{align} \label{def_d}
D^g_\varepsilon(\tau) &\equiv \{\mathbf{w} \in H^m(\mathbb S^3; \, \mathbb R^{N_g}) \, : \, \| \mathring{\mathbf{u}}_g (\tau) - \mathbf{w} \|_{H^m(\mathbb S^3; \, \mathbb R^N)} \leq \varepsilon \}, \\
D^{\mathfrak f}_{\varepsilon, \delta} &\equiv \{\mathbf{v} \in H_0^m(\mathbb S^3 \times \Omega_{1/4}; \, \mathbb R^{N_{\mathfrak f}}) \, :\label{def_df}\\
& \qquad \, \mathrm{supp}(\mathbf{v}) \subset \mathbb S^3 \times \Omega_\delta, \| \mathbf{v}\|_{H_0^m(\mathbb S^3 \times \Omega_{1/4}; \, \mathbb R^{N_{\mathfrak f}})} \leq \varepsilon\}. \nonumber 
\end{align}
\end{subequations}
Note that $D_\varepsilon^g(\tau)$ and $D_{\varepsilon,\delta}^{\mathfrak f}(\tau)$ are bounded, closed subsets of $H^m(\mathbb S^3; \, \mathbb R^{N_g})$ and $H^m(\mathbb S^3 \times \Omega_{1/4}; \mathbb R^{N_{\mathfrak f}})$, respectively. Finally, denote
\begin{equation} \label{def_ded}
\bm D_{(\varepsilon, \delta)}(\tau) \equiv D^g_\varepsilon(\tau) \times D^{\mathfrak f}_{\varepsilon, \delta},
\end{equation}
and
\begin{equation}
D_\varepsilon^g \equiv \bigcup_{\tau\in [0, \pi]} D_\varepsilon^g(\tau) , \quad \bm D_{(\varepsilon, \delta)} \equiv \bigcup_{\tau\in [0, \pi]} \bm D_{(\varepsilon, \delta)}(\tau)
\end{equation}
The initial data $\mathring{\mathbf{u}}$ will be chosen in $\bm D_{(\varepsilon, 1/2)}(0)$.

\subsection{Application of Kato's theorem}

In the following lemmas we establish the necessary conditions on the operators (\ref{operators_2})--(\ref{operators_4}) so that Kato's theorem can be applied. It has already been discussed in Proposition \ref{prop_cef_sh} that the coefficients matrices (\ref{operators_1}) of the system (\ref{cefmvm1}) of the geometric fields fulfil all necessary conditions.

\begin{lemma} \label{lem_cond_kato} (Assumptions of Kato's theorem -- matter equations) \\
The operators (\ref{operators_2}), (\ref{operators_3}) in the system for the matter fields $\mathbf{u}_{\mathfrak f}$ fulfil the following conditions:
\begin{enumerate}[(i)]
\item There exists a constant $\vartheta > 0$ such that for all $\nu, \bma = 0,\dots, 3$ one has the bounds, for all $\mathbf{v} \in D_\varepsilon^g$ and $\mathbf{w} \in D_{(\varepsilon, 1/4)}^{\mathfrak f}$,
\begin{subequations}
\begin{align}
\left\| \bm A_{\mathfrak f}^\nu[\mathbf{v}] \right\|_{H^m(\mathbb S^3 \times \Omega_{1/4}; \, \mathbb R^{N_{\mathfrak f} \times N_{\mathfrak f}})} + \left\| \mathfrak A_\bma[\mathbf{v}] \right\|_{H^m_{\mathrm{ul}}(\mathbb S^3 \times \Omega_{1/4}; \, \mathbb R^{N_{\mathfrak f} \times N_{\mathfrak f}})} &\leq \vartheta, \\
\left\| \bm F_{\mathfrak f}[\mathbf{v}, \mathbf{w}] \right\|_{H_0^m(\mathbb S^3 \times \Omega_{1/4}; \, \mathbb R^{N_{\mathfrak f}})} &\leq \vartheta.
\end{align}
\end{subequations}
\item There exists a (Lipschitz) constant $\mu > 0$ such that for all $j = 1,2,3$, $\bma = 0,\dots, 3$ one has the bounds, for all $ \mathbf{v}_1, \mathbf{v}_ 2 \in D_\varepsilon^g$ and $ \mathbf{w}_1, \mathbf{w}_2 \in D_{(\varepsilon, 1/4)}^{\mathfrak f}$,
\begin{subequations}
\begin{align}
\left\| \bm A_{\mathfrak f}^j[\mathbf{v}_1] - \bm A_{\mathfrak f}^j[\mathbf{v}_2]  \right\|_{H^m(\mathbb S^3 \times \Omega_{1/4}; \, \mathbb R^{N_{\mathfrak f} \times N_{\mathfrak f}})} &\leq \mu \left\| \mathbf{v}_1 - \mathbf{v}_ 2 \right\|_{H^m(\mathbb S^3; \, \mathbb R^{N_{g}})}, \\
\left\| \mathfrak A_\bma[\mathbf{v}_1] - \mathfrak A_\bma[\mathbf{v}_2] \right\|_{H^m(\mathbb S^3 \times \Omega_{1/4}; \, \mathbb R^{N_{\mathfrak f} \times N_{\mathfrak f}})} &\leq \mu \left\| \mathbf{v}_1 - \mathbf{v}_ 2 \right\|_{H^m(\mathbb S^3; \, \mathbb R^{N_{g}})},
\end{align}
\begin{multline}
\left\| \bm F_{\mathfrak f}[\mathbf{v}_1, \mathbf{w}_1] - \bm F_{\mathfrak f}[\mathbf{v}_2, \mathbf{w}_2]  \right\|_{H_0^m(\mathbb S^3 \times \Omega_{1/4}; \, \mathbb R^{N_{\mathfrak f}})} \\ \leq \mu \big(\left\| \mathbf{v}_1 - \mathbf{v}_ 2 \right\|_{H^m(\mathbb S^3; \, \mathbb R^{N_{g}})} + \left\| \mathbf{w}_1 - \mathbf{w}_ 2 \right\|_{H_0^m(\mathbb S^3 \times \Omega_{1/4}; \, \mathbb R^{N_{\mathfrak f}})} \big).
\end{multline}
\end{subequations}
\item There exists a (Lipschitz) constant $\mu' > 0$ such that, for all $\mathbf{v}_1, \mathbf{v}_ 2 \in D_\varepsilon^g$,
\begin{equation}
\left\| \bm A_{\mathfrak f}^\nu[\mathbf{v}_1] - \bm A_{\mathfrak f}^\nu[\mathbf{v}_2]  \right\|_{H^{m-1}(\mathbb S^3 \times \Omega_{1/4}; \, \mathbb R^{N_{\mathfrak f} \times N_{\mathfrak f}})} \leq \mu' \left\| \mathbf{v}_1 - \mathbf{v}_ 2 \right\|_{H^{m-1}(\mathbb S^3; \, \mathbb R^{N_{\mathfrak g}})},
\end{equation}
\item For each $\nu = 0,\dots, 3$ the matrices $\bm A^\nu_{\mathfrak f}[\bmv](\underline x, v_\bmb)$ and $\mathfrak A_\bma[\bmv](\underline x, v_\bmb)$ are symmetric for all $\underline x \in \mathbb S^3$, $v_\bmb \in \Omega_{1/4}$.
\item There exists a constant $d>0$ such that, for all $\mathbf{v}, \mathbf{v}_1, \mathbf{v}_ 2 \in D_\varepsilon^g$ and $\mathbf{w}, \mathbf{w}_1, \mathbf{w}_2 \in D_{(\varepsilon, 1/4)}^{\mathfrak f}$,  $\bm A_{\mathfrak f}^0[\mathbf{v}](\underline x, v_\bmb) \geq d$ for all $\underline x \in \mathbb S^3$, $v_\bmb \in \Omega_{1/4}$.
\end{enumerate}
\end{lemma}

\begin{proof}
By inspecting the formulas of the operators one can easily convince oneself that the continuity and symmetry properties hold. For boundedness, it is essential that the $v$-variables do not range over $\mathbb R^4$ but only over the bounded domain $\Omega_{1/4}$.
\end{proof}

For the source term in equation (\ref{operators_3}) of the system (\ref{cefmvm1}) for the geometric fields only the source terms containing the matter fields have to be considered. For the other quantities Proposition \ref{prop_cef_sh} provides the necessary properties.

\begin{lemma} (Assumptions of Kato's theorem -- metric equations) \label{lem_regular_source} \\
Consider $\mathbf{u}_g^\star, \mathbf{u}_g^{\star (1)}, \mathbf{u}_g^{\star (2)} \in D_\varepsilon^g$ and $\mathbf{u}_{\mathfrak f}^\star, \mathbf{u}_{\mathfrak f}^{\star (1)}, \mathbf{u}_{\mathfrak f}^{\star (2)} \in D_{(\varepsilon, 1/4)}^{\mathfrak f}$ (defined in (\ref{def_ded})), where
\begin{equation}
\mathbf{u}_g^\star = \left(\dots, (\Gamma_\bma{}^\bmc{}_\bmb)^\star, \dots \right), \quad \mathbf{u}_{\mathfrak f}^\star = \left(\mathfrak f_\star, \mathfrak f_\bma^\star, \varphi^\bmb_\star \right).
\end{equation}

Consider the operators
\begin{subequations}
\begin{align}
\mathcal T_{\bma \bmb}[\mathbf{u}_{\mathfrak f}^\star](\underline x) &= 8\pi \int_{\mathbb R^3} \mathfrak f_\star(\underline x, v_\bmc) \frac{v_\bma v_\bmb}{\sqrt{v_1^2 + v_2^2 + v_3^2}} \, \mathrm dv_1 \mathrm dv_2 \mathrm dv_3, \\
\mathcal T_{\bma \bmb \bmc}[\mathbf{u}_g^\star, \mathbf{u}_{\mathfrak f}^\star](\underline x) &= 8\pi (\Gamma_\bma{}^\bme{}_\bmd)^\star(\underline x)\, \int_{\mathbb R^3} \varphi^\bmd_\star(0, \underline x, v^\bma) \,  \frac{v_\bme v_\bmb v_\bmc}{|v|} \, \mathrm dv_1 \mathrm dv_2 \mathrm dv_3.
\end{align}
\end{subequations}
These operators fulfil the conditions:
\begin{enumerate}[(i)]
\item There exists a constant $\vartheta > 0$ such that, for all $\bma, \bmb, \bmc, \bmd, \bme = 0,\dots, 3$, the following bounds hold: for all $\mathbf{u}_g^\star\in D_\varepsilon^g$ and $\mathbf{u}_{\mathfrak f}^\star \in D_{(\varepsilon, 1/4)}^{\mathfrak f}$,
\begin{subequations}
\begin{align}
\left\| \mathcal T_{\bma\bmb}[\mathbf{u}_{\mathfrak f}^\star] \right\|_{H^m(\mathbb S^3; \, \mathbb R^{N_g})} &\leq \vartheta, \\
\left\| \mathcal T_{\bmc\bmd\bme}[\mathbf{u}_g^\star, \mathbf{u}_{\mathfrak f}^\star] \right\|_{H^m(\mathbb S^3; \, \mathbb R^{N_g})} &\leq \vartheta.
\end{align}
\end{subequations}
\item There exists a (Lipschitz) constant $\mu > 0$ such for that, all $\bma, \bmb, \bmc, \bmd, \bme = 0,\dots, 3$, the following bounds hold: for all $ \mathbf{u}_g^{\star (1)}, \mathbf{u}_g^{\star (2)} \in D_\varepsilon^g$ and $\mathbf{u}_{\mathfrak f}^{\star (1)}, \mathbf{u}_{\mathfrak f}^{\star (2)} \in D_{(\varepsilon, 1/4)}^{\mathfrak f}$
\begin{subequations}
\begin{equation}
\left\| \mathcal T_{\bma\bmb}[\mathbf{u}_{\mathfrak f}^{\star(1)}] - \mathcal T_{\bma\bmb}[\mathbf{u}_{\mathfrak f}^{\star(2)}] \right\|_{H^m(\mathbb S^3; \, \mathbb R^{N_g})} \leq \mu \left\| \mathbf{u}_{\mathfrak f}^{\star(1)} - \mathbf{u}_{\mathfrak f}^{\star(2)} \right\|_{H_0^m(\mathbb S^3 \times \Omega_{1/4}; \, \mathbb R^{N_{g}})}, \label{lipschitz1}
\end{equation}
\begin{multline}
\left\| \mathcal T_{\bmc\bmd\bme}[\mathbf{u}_g^{\star(1)}, \mathbf{u}_{\mathfrak f}^{\star(1)}] - \mathcal T_{\bmc\bmd\bme}[\mathbf{u}_g^{\star(2)}, \mathbf{u}_{\mathfrak f}^{\star(2)}] \right\|_{H^m(\mathbb S^3; \, \mathbb R^{N_g})} \\
\leq \mu \left(\left\| \mathbf{u}_{g}^{\star(1)}-\mathbf{u}_{g}^{\star(2)} \right\|_{H^m(\mathbb S^3; \, \mathbb R^{N_{g}})}+ \left\| \mathbf{u}_{\mathfrak f}^{\star(1)} - \mathbf{u}_{\mathfrak f}^{\star(2)} \right\|_{H_0^m(\mathbb S^3 \times \Omega_{1/4}; \, \mathbb R^{N_{\mathfrak f}})} \right).
\end{multline}
\end{subequations}
\end{enumerate}
\end{lemma}

\begin{proof}
Let $\mathbf{u}_g^\star, \mathbf{u}_g^{\star (1)}, \mathbf{u}_g^{\star (2)} \in D_\varepsilon^g$ and $\mathbf{u}_{\mathfrak f}^\star, \mathbf{u}_{\mathfrak f}^{\star (1)}, \mathbf{u}_{\mathfrak f}^{\star (2)} \in D_{(\varepsilon, 1/4)}^{\mathfrak f}$, arbitrary.
\begin{enumerate}[(i)]
\item Boundedness of the operators can be seen as follows. Observe first that in the support of $\mathbf{w}$ we have $1/4 < |v| < 4$. So, the terms $v_\bma v_\bmb/|v|$ are bounded and can be pulled out of the integrals.  Furthermore, since $\mathbf{u}_g^\star, \mathbf{u}_g^{\star (1)}, \mathbf{u}_g^{\star (2)} \in D_\varepsilon^g$, the Christofell symbols of the perturbation are bounded. 
Integrating over $\mathbb S^3$ and using Jensen's inequality yield the $L^2(\mathbb S^3 \times \Omega_{1/4}; \, \mathbb R)$-norm of $\mathfrak f$ times a constant. Note that there the compact support of $\mathfrak f$ in velocities is crucial. The $L^2(\mathbb S^3\times \Omega_{(1/4)}; \mathbb R)$-norm of $\mathfrak f^\star$ can course be bounded by the $H_0^m(\mathbb S^3\times \Omega_{(1/4)}; \mathbb R)$-norm of $\mathfrak f^\star$, which in turn is smaller than $\varepsilon$.

\item Lipschitz-continuity with respect to $\mathbf{w}$. Consider first (\ref{lipschitz1}). Again, we observe that due to the assumption on the support of $w$ we can pull the factors $v_\bma v_\bmb/|v|$ out of the integrals. Now the claim follows by linearity of the integral. The operator $\mathcal T_{\bma\bmb\bmc}$ can be dealt with analogously. 
\end{enumerate}
\end{proof}

We are now ready to state and prove our stability result.

\begin{theorem} \label{main_theorem}
Let $\varepsilon > 0$ and consider initial data $\mathbf{u}_\star \in \bm D_{\varepsilon, 1/4}$, as defined in (\ref{def_ded}), on $\mathbb S^3$. Given $m\geq 5$ and $\tau_\bullet > \frac{3}{4}\pi$ then, if $\varepsilon$ is small enough this initial data give rise to a solution
\begin{multline}
\mathbf{u} \in  C^0 \left([0,\tau_\bullet]; H^{m}(\mathbb S^3; \mathbb R^{N_g}) \times H^{m}(\mathbb S^3 \times \mathbb R_v^4; \mathbb R^{N_{\mathfrak f}})\right) \\ \cap C^1 \left([0,\tau_\bullet]; H^{m-1}(\mathbb S^3; \mathbb R^{N_g}) \times H^{m-1}(\mathbb S^3 \times \mathbb R_v^4; \mathbb R^{N_{\mathfrak f}})\right)
\end{multline}
of the conformal Einstein field equations with massless Vlasov matter, equations (\ref{SpinorialZQ1})--(\ref{SpinorialZQ6}), (\ref{conf_f_1})--(\ref{conf_f_3}). This solution satisfies
\begin{equation}
\mathrm{supp}(\mathbf{u}_{\mathfrak f}(\tau, \cdot)) \in \mathbb S^3 \times \left\{(v_0, \dots, v_3) \in \mathbb R_v^4 \, : \, v_0 = -|v|, \frac 1 4 \leq |v| \leq 2 \right\}
\end{equation}
for all $\tau \leq \tau_\bullet$. Furthermore, given a sequence of initial data $\mathbf{u}_\star^{(n)} = \mathring{\mathbf{u}}_\star^{(n)} + \breve{\mathbf{u}}^{(n)}_\star$ such that
\[
\parallel \breve{\mathbf{u}}_\star^{(n)}\parallel_m <\varepsilon \qquad  \mbox{and} \qquad \breve{\mathbf{u}}^{(n)}_\star\rightarrow 0 \qquad \mbox{as} \qquad n\rightarrow \infty,
\]
then for the corresponding solutions $\breve{\mathbf{u}}^{(n)}_n$ one has that $\breve{\mathbf{u}}^{(n)}\rightarrow 0$ uniformly in $\tau\in[0,\tau_\bullet]$ as $n\rightarrow \infty$.
\end{theorem}

\begin{remark}
{\em The proof of Theorem \ref{main_theorem} relies on the proofs of \cite[Theorems II and III]{kato}. Our situation differs in in the following points from the situation considered there:
\begin{enumerate}[(i)]
\item In the current setting we look for an existence and stability result for a system of two coupled, non-linear symmetric hyperbolic systems, one system for the geometric fields $\bm u_g$ and one system for the matter fields $\bm u_{\mathfrak f}$. The two systems cannot simply be considered as one single system since the unknowns are of different dimensions.
\item The domain of the solution functions of the symmetric hyperbolic systems considered in this article is not $\mathbb R^m$ for some $m > 0$ but $\mathbb S^3$ or $\mathbb S^3 \times \Omega_{1/4}$. A localisation procedure described in \cite{MR1131434} or \cite{VKbook} makes it possible to state and prove Kato's existence and stability theorems on compact manifolds without boundary (in particular on $\mathbb S^3$). For solutions that are launched by initial data which are sufficiently close to the vacuum solution and the matter fields of which have $v$-support contained in $\Omega_{1/2}$, we will show that the support in the $v$-variables stays within $\Omega_{1/4}$. For this reason, for the specific system of equations at hand and solutions close to the de Sitter solution we can work with $\Omega_{1/4}$ instead of $\mathbb R^4$.
\end{enumerate}}
\end{remark}

\begin{proof}[Proof of Theorem \ref{main_theorem}]
In this proof we will use the shorthand
\begin{align}
H^m &\equiv H^m(\mathbb S^3; \, \mathbb R^{N_g}) \times H^m(\mathbb S^3\times \Omega_{1/4}; \, \mathbb R^{N_{\mathfrak f}} ), \\
\|\cdot \|_m &\equiv \|\cdot \|_{H^m}.
\end{align}
Recall that $N_g$ denotes the number of independent components of the geometric fields and $N_{\mathfrak f}$ the number of independent components of the matter fields. $\Omega_{1/4}$ is defined in (\ref{def_dom_omega}). \par

The proof consists in several steps. First a {\em local existence result} for the solution $\mathbf{u}$ with initial data $\mathbf{u}_\star$ is demonstrated on a (short) time interval $[0,T']$, $T' \leq \tau_\bullet$ via a contraction argument \cite[Proof of Thm.~II]{kato}. This proof implies that for initial data $\tilde{\mathbf{u}}_\star$ close to $\mathbf{u}_\star$ a solution $\tilde{\mathbf{u}}$ exists on the same interval $[0,T']$. In the next step a {\em local stability result} is established on the (short) interval $[0,T']$, i.e.~$\sup_{0\leq \tau \leq T'} \|\tilde{\mathbf{u}}(\tau) - \mathbf{u}(\tau)\|_m \to 0$, uniformly in $\tau$, as $\|\tilde{\mathbf{u}}_\star - \mathbf{u}_\star \|_m \to 0$ \cite[Proof of Thm.~III]{kato}. Finally these local results can be applied successively in order to establish them on the interval $[0,\tau_\bullet]$ \cite[p.~200, last paragraph]{kato}.

\begin{remark}
{\em This proof relies on the proofs of Theorems II and III in \cite{kato} and details are given only for the parts which differ from these proofs. Since the reader might want to compare with the proofs in \cite{kato}, we comment on the notations. In \cite{kato}, the contraction for the local existence result is set up on a set $\bm S$ of functions which stay close to a function $u_{00}$. This is a technical help function, needed to use stability estimates for linear hyperbolic systems which need more regularity. In our setting the background solution $\mathring {\bm u}(\tau)$ is chosen for $u_{00}$, with $\tau = 0$ for the result on $[0,T']$, and $\tau >0$ in the successive application of the local result. At this point, the higher regularity of the initial data of the background solution is used. }
\end{remark}

First we establish the local existence result.

\begin{lemma}
Let $\varepsilon > 0$ and consider initial data $\mathbf{u}_\star \in \bm D_{(\varepsilon, 1/4)}(0)$ such that $\|  \mathbf{u}_\star - \mathring{\mathbf{u}}_\star \|_m \leq \varepsilon_2$ for some $\varepsilon_2 \in (0,\varepsilon]$. Then, if $\varepsilon_2$ is chosen sufficiently small, there exists $R$ such that
\begin{itemize}
\item if $(\mathbf{w}, \mathbf{v}) \in H^m$ and $\|(\mathbf{w}, \mathbf{v}) - \mathring{\mathbf{u}}_\star \|_m \leq R$ then $(\mathbf{w}, \mathbf{v}) \in \bm D_{\varepsilon,1/2}(0)$,
\item $\varepsilon_2 \left( \left\| \bm A^0_g[\mathbf{v}] \right\|_{L^\infty}^{1/2} + \left\| \bm A^0_{\mathfrak f}[\mathbf{v}] \right\|_{L^\infty}^{1/2} \right) \leq R/3$ for $(\mathbf{w}, \mathbf{v}) \in \bm D_{\varepsilon,1/2}(0).$
\end{itemize}
\end{lemma}

We define the operator
\begin{multline}
\Phi: C^1\left([0,T]; \,H^m(\mathbb S^3; \mathbb R^{N_g}) \times H_0^m(\mathbb S^3 \times \Omega_{1/4} ; \mathbb R^{N_{\mathfrak f}})\right) \\ \to C^0\left([0,T]; \,H^m(\mathbb S^3; \mathbb R^{N_g}) \times H^m(\mathbb S^3 \times \mathbb R^4 ; \mathbb R^{N_{\mathfrak f}})\right)
\end{multline}
which assigns to a function $(\mathbf{w}, \mathbf{v})$ the solution $\mathbf{u}^{(\mathbf{w}, \mathbf{v})} \equiv (\mathbf{u}^{(\mathbf{w}, \mathbf{v})}_g,\mathbf{u}^{(\mathbf{w}, \mathbf{v})}_{\mathfrak f})$ of the linear hyperbolic system
\begin{subequations}
\begin{multline}
\bm A_g^0[\mathbf{w}(\tau, \cdot)](\underline x) \partial_\tau  \mathbf{u}^{(\mathbf{w}, \mathbf{v})}_g + \bm A_g^i[\mathbf{w}(\tau, \cdot)](\underline x) \partial_{x^i}  \mathbf{u}^{(\mathbf{w}, \mathbf{v})}_g \\= \bm F_g[\mathbf{w}(\tau, \cdot), \mathbf{v}(\tau, \cdot)](\underline x), \label{lin_loc_1}
\end{multline}
\begin{multline}
\bm A_{\mathfrak f}^0[\mathbf{w}(\tau, \cdot)](\underline x, v_\bma) \partial_{\tau}  \mathbf{u}^{(\mathbf{w}, \mathbf{v})}_{\mathfrak f} + \bm A_{\mathfrak f}^i[\mathbf{w}(\tau, \cdot)](\underline x, v_\bma) \partial_{x^i}  \mathbf{u}^{(\mathbf{w}, \mathbf{v})}_{\mathfrak f} \\
+ \mathfrak A_\bmc[\mathbf{w}(\tau, \cdot)](\underline x, v_\bma) \partial_{v^\bmc} \mathbf{u}^{(\mathbf{w}, \mathbf{v})}_{\mathfrak f}  = \bm F_{\mathfrak f}[\mathbf{w}(\tau, \cdot),  \mathbf{v}(\tau, \cdot)](\underline x, v_\bma) \label{lin_loc_2} 
\end{multline}
\end{subequations}
equipped with the initial data $\mathbf{u}_\star \in \bm D_{(\varepsilon, 1/4)}$. Note that the linear system (\ref{lin_loc_1})--(\ref{lin_loc_2}) is not coupled. So we can solve each system individually by Theorem I in \cite{kato}. This shows that the operator $\Phi$ is well-defined. \par
Next, we wish to set up a contraction argument. To this end we first define the set $\bm S = \bm S(R,T',L')$ as the set of all functions $(\mathbf{w}, \mathbf{v}) : [0,T'] \to H^m$ such that
\begin{align}
\| (\mathbf{w}, \mathbf{v})(\tau, \cdot) - \mathring{\mathbf{u}}_{\star}\|_m &\leq R,\quad \text{for}\, \tau \in [0,T'], \label{cond_s_1} \\
\| (\mathbf{w}, \mathbf{v})(\tau, \cdot) - (\mathbf{w}, \mathbf{v})(\tau', \cdot) \|_{m-1} &\leq L' (\tau-\tau') \label{cond_s_2}\\
&\text{for}\,0\leq \tau' \leq \tau \leq T' \nonumber
\end{align}
where $T' \leq T$ and $L'$ are positive constants to be determined later. \par
The next steps are now to check that if $T'$ is chosen sufficiently small and $L'$ sufficiently large we have for all $(\mathbf{w}, \mathbf{v}) \in \bm S$ that $\Phi (\mathbf{w}, \mathbf{v}) \in \bm S$, and that $\Phi$ acts as a contraction on $\bm S$. First we check that $\Phi (\mathbf{w}, \mathbf{v}) \in \bm S$. We have to verify (\ref{cond_s_1}), (\ref{cond_s_2}), and the support condition
\begin{equation} \label{cond_supp}
\mathrm{supp}(\Phi(\mathbf{w}, \mathbf{v})(\tau,\cdot)) \subset \mathbb S^3 \times \Omega_{1/4}, \quad \tau \in [0,T'].
\end{equation}
\cite[Theorem I]{kato} provides the stability estimates
\begin{align}
\left\| \mathbf{u}^{(\mathbf{w}, \mathbf{v})}_g(\tau,\cdot) - \mathring{\mathbf{u}}_g^\star\right\|_{H^m(\mathbb S^3; \, \mathbb R^{N_g})} \leq \beta(L', T'), \\
\left\| \mathbf{u}^{(\mathbf{w}, \mathbf{v})}_{\mathfrak f}(\tau, \cdot) - \mathring{\mathbf{u}}_{\mathfrak f}^\star\right\|_{H^m(\mathbb S^3 \times \mathbb R^4; \, \mathbb R^{N_{\mathfrak f}})} \leq \beta(L',T'),
\end{align}
and
\begin{align}
\left\| \partial_\tau \mathbf{u}^{(\bm w, \bm v)}_g(\tau,\cdot) \right\|_{H^{m-1}(\mathbb S^3; \, \mathbb R^{N_g})} &\leq c(1+\vartheta^{m-1}) \vartheta (1+ \|\mathring{\mathbf{u}}_\star\|_m + \beta(L',T')), \\
\left\| \partial_\tau \mathbf{u}^{(\bm w, \bm v)}_g(\tau,\cdot) \right\|_{H^{m-1}(\mathbb S^3 \times \mathbb R^4; \, \mathbb R^{N_{\mathfrak f}})} &\leq c(1+\vartheta^{m-1}) \vartheta (1+ \|\mathring{\mathbf{u}}_\star\|_m + \beta(L',T')),
\end{align}
where $c>0$ is a universal constant and, furthermore,
\begin{align}
\beta(L',T') &= c \vartheta e^{\alpha(L') T'} \left( \| \mathbf{u}_\star - \mathring{\mathbf{u}}_\star \|_m + c (1+\vartheta^m) \vartheta (1+\|\mathring{\mathbf{u}}_\star \|_{m+1}) T' \right), \\
\alpha(L') &= c(\nu + \mu'L' + \vartheta + \vartheta^{m+1}),
\end{align}
and the constants $\vartheta$, $\mu$, and $\mu'$ are the constants of Lemma \ref{lem_cond_kato}.
Since
\begin{equation} \label{norms}
\|\cdot\|_m \lesssim \|\cdot\|_{H^m(\mathbb S^3; \, \mathbb R^{N_g})} + \|\cdot\|_{H^m(\mathbb S^3 \times \mathbb R^4; \,\mathbb R^{N_{\mathfrak f}})}
\end{equation}
the conditions (\ref{cond_s_1}) and (\ref{cond_s_2}) follow similarly as in [Kato, 1975]. \par
The support property (\ref{cond_supp}) can be shown by an analysis of the characteristic system of the Vlasov equation. Let $\varsigma \mapsto (x^\mu(\varsigma), v_\bma(\varsigma)$, $\varsigma \in \mathbb R$ be a solution of the characteristic system of the massless Vlasov equation reading
\begin{subequations}
\begin{align}
&\dot x^\mu(\varsigma) = \eta^{\bma \bmb} \overset{\scriptsize{\mathbf{w}}}{e_\bma{}^\mu}(x(\varsigma)) \, v_\bmb(\varsigma), \label{char_sys_1} \\
&\dot v_\bmd(\varsigma) = \eta^{\bma\bmb} \overset{\mathbf{w}}{\Gamma_\bmb{}^\bmc{}_\bmd}(x(\varsigma))\, v_\bma(\varsigma) v_\bmc(\varsigma),
\end{align}
\end{subequations}
where the mass shell condition
\begin{equation}
v_\bmzero(\varsigma) = -\sqrt{\left(v_{\bm 1}(\varsigma)\right)^2 + \left(v_{\bm 2}(\varsigma)\right)^2 + \left(v_{\bm 3}(\varsigma)\right)^2} \label{char_sys_2}
\end{equation}
is propagated. The notation $\overset{\scriptsize{\mathbf{w}}}{e_\bmc{}^\mu}$, $\overset{\mathbf{w}}{\Gamma_\bma{}^\bmc{}_\bmb}$ indicates that we refer to the components of $\mathbf{w}$ which correspond to $e_\bmc{}^\mu$ and $\Gamma_\bma{}^\bmc{}_\bmb$, respectively. \par
The support of $f$ consists in characteristic curves $\varsigma \mapsto (x^\mu(\varsigma), v_\bmb(\varsigma))$ which are launched from the initial hypersurface $\mathbb S^3 \times \Omega_{1/2}$, which is characterised by $\tau = 0$, i.e.~$x^0(0) = 0$. Denote further
\begin{equation}
x_\star^\mu \equiv x^\mu(0), \quad v^\star_\bmb \equiv v_\bmb(0).
\end{equation}
The functions $\overset{\scriptsize{\mathbf{w}}}{e_\bmc{}^\mu}$, $\overset{\mathbf{w}}{\Gamma_\bma{}^\bmc{}_\bmb}$ can be written as
\begin{equation}
\overset{\mathbf{w}}{\Gamma_\bma{}^\bmc{}_\bmb} = \mathring \Gamma_\bma{}^\bmc{}_\bmb + \breve \Gamma_\bma{}^\bmc{}_\bmb, \qquad \overset{\scriptsize{\mathbf{w}}}{e_\bmc{}^\mu} = \mathring e_\bmc{}^\mu + \breve e_\bmc{}^\mu,
\end{equation}
the quantities with a ``$\mathring{\phantom{X}}$'' denote quantities of the background solution and quantities with ``$\breve{\phantom{X}}$'' denote a small perturbation. In view of the explicit expressions given in  (\ref{ds_quan1})--(\ref{ds_quan2}) for de Sitter or (\ref{min_quan_1})--(\ref{min_quan_2}) for Minkowski, this yields
\begin{equation}
\left|\overset{\mathbf{w}}{\Gamma_\bma{}^\bmc{}_\bmb}\right| \leq \epsilon_{\bmzero\bma}{}^\bmc{}_\bmb + \varepsilon, \qquad \left| \overset{\scriptsize{\mathbf{w}}}{e_\bmzero{}^\mu} \right| \geq \delta_0{}^\mu - \varepsilon, \qquad \left|\overset{\scriptsize{\mathbf{w}}}{e_\bmi{}^\mu} \right| \geq 1 - \varepsilon
\end{equation}
on $\mathbb{S}^3\times[0, T')$. Without loss of generality we assume $T' \leq 3\pi/4$. Otherwise we just change $T'$ to $3\pi/4$. Consider now the differential equation for $v_\bmzero$. Observing that $\mathring \Gamma_\bmb{}^\bmc{}_\bmzero = 0$ we deduce that there exists a constant $C>0$ such that
\[
\left|\dot v_\bmzero(\varsigma) \right| = \left|\eta^{\bma\bmb} \breve{\Gamma}_\bmb{}^\bmc{}_\bmzero (x(\varsigma))\, v_\bma(\varsigma) v_\bmc(\varsigma)\right| \leq C \varepsilon (v_\bmzero(\varsigma))^2.
\]
This equation yields
\begin{equation} \label{bound_v0}
\left| v_\bmzero (\varsigma)\right| \geq \left(\frac{1}{v^\star_\bmzero} + \varepsilon \varsigma\right)^{-1} \geq \left(2 + \varepsilon \varsigma\right)^{-1},
\end{equation}
where in the last inequality $|v^\star_\bmzero| \geq \frac 12$ has been used. Next, we control the range of the affine parameter $\varsigma$ for which the characteristic curve $(x^\mu(\varsigma), v_\bmb(\varsigma))$ reaches $x_\bmzero(\varsigma) = T'$. We have
\begin{equation}
\dot x^0(\varsigma) = \left| \eta^{\bma\bmb} \overset{\scriptsize{\mathbf{w}}}{e_\bma{}^\mu} v_\bmb \right| \geq (1 - \varepsilon) \left| v_\bmzero(\varsigma) \right| \geq \frac{1-\varepsilon}{2 + \varepsilon \varsigma}.
\end{equation}
In the second inequality we used $|v_\bmzero| \geq |v_{\bm 1}|, |v_{\bm 2}|, |v_{\bm 3}|$ and in the last inequality we substituted \eqref{bound_v0}. This yields
\begin{equation}
x^0(\varsigma) \geq \frac{1-\varepsilon}{\varepsilon} \ln\left(\frac{2+\varepsilon 2}{2}\right)
\end{equation}
and therefore see that $x^0(\varsigma)$ reaches $\min\{T', 3\pi/4\}$ before $\varsigma$ reaches
\begin{equation}
\hat \varsigma \equiv \frac{2}{\varepsilon} \left(e^{\frac{3\pi}{4}\frac{\varepsilon}{1-\varepsilon}} -1\right).
\end{equation}
Substituting this again into (\ref{bound_v0}) yields
\begin{equation}
\forall \varsigma\leq \hat \varsigma,\; \left| v_\bmzero(\hat \varsigma) \right| \geq  1- \dfrac12 e^{\frac{3\pi}{4} \frac{\varepsilon}{1-\varepsilon}} \geq \dfrac14,
\end{equation}
provided that $\varepsilon$ is small enough. Now we have established that for all $(\mathbf{w}, \mathbf{v})$ we have $\Phi(\mathbf{w}, \mathbf{v}) \in \bm S$. The next step is to show that $\Phi$ acts as a contraction on $\bm S$. This goes however analogously to the proof of \cite[Lemma 4.5]{kato}. By Banach's fixed point theorem, the local existence of the solution on $[0,T']$ follows. \par
The remaining steps of the proof require very little modification of the original proofs of \cite[Theorems II and II]{kato}. The operator $\Phi$ and the norms (\ref{norms}) have to be replaced. Furthermore, as above, results for linear hyperbolic systems have to be applied to each of the subsystems (\ref{lin_loc_1}) and (\ref{lin_loc_2}) separately. Finally, when the local result is applied successively in order to obtain the result on the interval $[0,\tau_\bullet]$, in the above, $\mathring{\bm u}_\star$ has to be replaced by $\mathring{\bm u}(T')$ etc. 
\end{proof}

\section{De Sitter like space-times} \label{sect_de_sitter}

In this section we state the non-linear stability result for the de Sitter space-time. It is obtained by applying Theorem \ref{main_theorem} to initial data where the background solution $\mathring{\bm u}$ is de Sitter space-time. \par
Let $\mathbf{u}_{\mathrm{dS}}$ be the collection of functions 

\begin{multline} \label{def_ug_ds}
\mathbf{u}_{\mathrm{dS}} \equiv \Big(\Xi_{\mathrm{dS}}, \; \left(\Sigma_{\bmA\bmA'}\right)_{\mathrm{dS}}, \; s_{\mathrm{dS}}, \; \left(e_{\bmA\bmA'}{}^a\right)_{\mathrm{dS}}, \\ \left(\Gamma_{\bmA\bmA'\bmB\bmC}\right)_{\mathrm{dS}}, \; \left(\Phi_{\bmA\bmA'\bmB\bmB'}\right)_{\mathrm{dS}}, \; \left(\phi_{\bmA\bmB\bmC\bmD}\right)_{\mathrm{dS}} \Big),
\end{multline}

where
\begin{equation} \label{ds_quan1}
\Xi_{\mathrm{dS}} = \cos(\tau), \quad s_{\mathrm{dS}} = - \frac 1 4 \cos(\tau), \quad \left(e_{\bmA\bmA'}{}^a\right)_{\mathrm{dS}} = \sigma_{\bmA\bmA'}{}^a,
\end{equation}
with $\sigma_{\bmA\bmA'}{}^a$ the Infeld-van der Waerden symbols and the remaining functions are the spinorial counterparts of
\begin{equation} \label{ds_quan2}
\begin{aligned}
\left( \Gamma_\bma{}^\bmc{}_\bmb\right)_{\mathrm{dS}} &= \epsilon_{0\bma}{}^\bmc{}_\bmb, & \left( \Sigma_\bmi \right)_{\mathrm{dS}} &= 0, \\
\left(L_{\bma\bmb}\right)_{\mathrm{dS}} &= \delta_\bma{}^0 \delta_\bmb{}^0 - \frac 12 \eta_{\bma\bmb}, & \left(d^\bma{}_{\bmb\bmc\bmd}\right)_{\mathrm{dS}} &= 0.
\end{aligned}
\end{equation}
We are now using \cite[Lemma 15.1]{VKbook}.
\begin{lemma} \label{lem_background}
The functions $\mathbf{u}_{\mathrm{dS}}$ defined over the Einstein cylinder $\mathbb R \times \mathbb S^3$ constitute a solution to the spinorial vacuum conformal Einstein field equations, where the gauge source functions associated to these solutions are given by (\ref{gsf}).
\end{lemma}

The  de Sitter spacetime corresponds to the domain 
\begin{equation}
\tilde{\mathcal M}_{\mathrm{dS}} = \{(\tau,\underline{x}) \in \mathscr E \, : \, \tau \in (-\pi,\pi), \, \underline{x}\in\mathbb{S}^3\}.
\end{equation}
In particular, observe that $\tilde{\mathcal M}_{\mathrm{dS}}$ is a domain in $\mathscr E$ where $\Xi_{\mathrm{dS}} > 0$. Furthermore, we denote 
\begin{equation}
\mathbf{u}_{\mathrm{dS}}^\star \equiv \mathbf{u}_{\mathrm{dS}} \big|_{\tau = 0}.
\end{equation}
Hence,  $\mathbf{u}_{\mathrm{dS}}^\star$ is initial data which, if prescribed on $\mathbb S^3$, yields $\mathbf{u}_{\mathrm{dS}}$ on $[0,\infty)\times \mathbb S^3$. \par
The perturbation ansatz $\bm g_{\mathrm{dS}} + \breve{\bm g}$ of the de Sitter metric $\bm g_{\mathrm{dS}}$ on the Einstein cylinder $\mathscr E$ gives rise to initial data $\mathbf{u}^\star$ on $\mathbb S^3$. It turns out (cf.~\cite[Section 15.2]{VKbook}) that this initial data is of the form $\mathbf{u}^\star = \mathbf{u}_{\mathrm{dS}} + \breve{\mathbf{u}}^\star$ and that $\mathbf{u}^\star \in \bm D_{(\varepsilon, 1/2)}$, if $\breve{\bm g}$ is small enough, where in the definition (\ref{def_ded}) of $\bm D_{(\varepsilon, 1/2)}$ we choose $\mathbf{u}_{\mathrm{dS}}^\star$ as background solution $\mathring{\mathbf{u}}^\star$. We assume that the initial data $\mathbf{u}^\star$ solves the conformal constraint equations on $\mathbb S^3$ with source terms generated by the initial value $\mathbf{u}_{\mathfrak f}^\star$ of the particle distribution function. Then Theorem \ref{main_theorem} yields a solution $\mathbf{u}$ (which can be extended) on $[0, \tau_\bullet] \times \mathbb S^3 \times \mathbb R_v^4$ which is close to the de sitter solution $ \mathbf{u}_{\mathrm{dS}}$, where $\tau_\bullet \geq 3\pi/4$. \par
Initially, the conformal factor $\Xi$ is not zero (it is close to one to be precise). We will now demonstrate that there exists a hypersurface $\mathscr I^+$, diffeomorphic to $\mathbb S^3$ such that $\Xi = 0$ on $\mathscr I^+$ and $\Xi > 0$ on the domain bounded by $\{ \tau = 0\}$ and $\mathscr I^+$. This domain will then be identified as the perturbed de Sitter space-time.  This characterisation is now made precise.

\begin{lemma} \label{Lemma:DeSitterConformalBoundary}
Given a solution $\breve{\mathbf{u}}$ as given by Theorem \ref{main_theorem} with $\| \breve{\mathbf{u}}_\star\|_m < \varepsilon$ sufficiently small, there exists a function $\tau_+=\tau_+(\underline{x})$, $\underline{x}\in \mathbb{S}^3$ such that $0<\tau_+(\underline{x})<\tau_\bullet$ and
\begin{eqnarray*}
&\Xi >0 \quad  \mbox{on} \quad  \tilde{\mathcal{M}} \equiv \big\{ (\tau,\underline{x}) \in [0,\tau_\bullet] \times \mathbb{S}^3 \;|\; 0\leq \tau < \tau_+(\underline{x})\big\}, & \\
& \Xi =0 \quad \mbox{and} \quad \Sigma_a\Sigma^a = -\frac{1}{3}\lambda<0 \quad \mbox{on} \quad  \mathscr{I}^+ \equiv \big\{ (\tau_+(\underline{x}),\underline{x})\in \mathbb{R}\times \mathbb{S}^3 \big\}. &
\end{eqnarray*}
\end{lemma}

\begin{remark}
\label{Remark:GeodesicCompleteness}
{\em From the previous lemma, it follows that the solution to the Einstein-Vlasov equations obatined from Theorem \ref{main_theorem} is future asymptotically simple. Accordingly, the spacetime is null geodesically complete  ---see e.g. \cite{MR944085}. The future timelike geodesic completeness of the solution can be obtained using the notion of conformal geodesics as described in \cite{MR2911223}. Conformal geodesics provide a convenient conformal description of physical geodesics. The completeness of these curves in the physical portion of the conformal spacetime follows from the Cauchy stability of solutions to ordinary differential equations and the fact that is possible to obtain a common $\epsilon$ for all the curves starting on the initial hypersurface as it is compact.  }
\end{remark}

\begin{theorem} \label{result_de_Sitter}
Given $m \geq 5$,
a de Sitter-like solution $\mathbf{u}_\star = \bm u_{\mathrm{dS}}^\star + \breve{\mathbf{u}}_\star$ to the Einstein-Vlasov conformal constraint equations such that $|| \breve{\mathbf{u}}_\star ||_m < \varepsilon$ for $\varepsilon>0$ suitably small gives rise to a unique $C^{m-2}$ solution to the conformal Einstein-Vlasov field equations on
\[
\mathcal{M} \equiv \tilde{\mathcal{M}} \cup \mathscr{I}^+
\]
with $\tilde{\mathcal{M}}$ and $\mathscr{I}^+$ as defined in Lemma \ref{Lemma:DeSitterConformalBoundary}. The solution implies, in turn, a solution $(\tilde{\mathcal{M}},\tilde{\bmg})$, to the Einstein field equations with de Sitter-like cosmological constant for which $\mathscr{I}^+$ represents conformal infinity. The space-time $(\tilde{\mathcal M}, \tilde{\bm g})$ is geodesically complete.
\end{theorem}

\begin{proof}
Once a solutions $\mathbf{u}$ on $[0,\tau_\bullet] \times \mathbb S^3$ is obtained by virtue of Theorem \ref{main_theorem}, the remaining properties can be shown by using that the geometric fields $\mathbf{u}_g$ are close to the de Sitter ones, $\bm u_{\mathrm{dS}}$. The arguments are identical to the vacuum case and we refer to \cite[Section 8]{VKbook}. 
\end{proof}

\section{Minkowski-like space-times} \label{sect_minkowski}

Theorem \ref{main_theorem} also allows to obtain a semi-global stability result for the Minkowski sace-time. To this end, we first discuss the background solution $\mathbf{u}_{\mathrm M}$ on the Einstein cylinder $\mathscr E$. Locally on $\mathbb S^3$ we work with the coordinates $(\psi, \vartheta, \varphi) \in [0,\pi] \times [0,\pi] \times [0,2\pi)$. Let $\mathbf{u}_{\mathrm{M}}$ be the collection of functions 
\begin{multline} \label{def_ug_mink}
\mathbf{u}_{\mathrm{M}} \equiv \Big(\Xi_{\mathrm{M}}, \; \left(\Sigma_{\bmA\bmA'}\right)_{\mathrm{M}}, \; \Sigma_{\mathrm{M}}, \\ \; \left(e_{\bmA\bmA'}{}^a\right)_{\mathrm{M}}, \; \left(\Gamma_{\bmA\bmA'\bmB\bmC}\right)_{\mathrm{M}}, \;\left(\Phi_{\bmA\bmA'\bmB\bmB'}\right)_{\mathrm{M}}, \; \left(\phi_{\bmA\bmB\bmC\bmD}\right)_{\mathrm{M}} \Big),
\end{multline}
where
\begin{equation} \label{min_quan_1}
\Xi_{\mathrm{M}} = \cos\tau + \cos \psi, \quad s_{\mathrm{M}} = -\frac 14 (\cos\tau - 3\cos\psi), \quad \left(e_{\bmA\bmA'}{}^a\right)_{\mathrm{M}} = \sigma_{\bmA\bmA'}{}^a,
\end{equation}
and $\sigma_{\bmA\bmA'}{}^a$ are the Infeld-van der Waerden symbols. The remaining functions in $\mathbf{u}_{\mathrm{M}}$ are the spinorial counterparts of
\begin{equation} \label{min_quan_2}
\begin{aligned}
\left( \Gamma_\bma{}^\bmc{}_\bmb\right)_{\mathrm{M}} &= \epsilon_{0\bma}{}^\bmc{}_\bmb, & \left( \Sigma_\bmi \right)_{\mathrm{M}} &= 0, \\
\left(L_{\bma\bmb}\right)_{\mathrm{M}} &= \delta_\bma{}^0 \delta_\bmb{}^0 - \frac 12 \eta_{\bma\bmb}, & \left(d^\bma{}_{\bmb\bmc\bmd}\right)_{\mathrm{M}} &= 0. 
\end{aligned}
\end{equation}

\begin{lemma}  \label{lem_min} The functions $\mathbf{u}_{\mathrm{M}}$ defined over the Einstein cylinder $\mathbb R \times \mathbb S^3$ constitute a solution to the spinorial vacuum conformal Einstein field equations, where the gauge source functions associated to these solutions are given by \eqref{gsf}.
\end{lemma}

In terms of the coordinates on  $\mathscr E$ the Minkowski spacetime corresponds to the domain 
\begin{equation}
\tilde{\mathcal M}_{\mathrm{M}} = \{x \in \mathscr E \, : \, 0< \psi <\pi, \, \psi -\pi < \tau <\pi-\psi\}.
\end{equation}
Recall that $\tilde{\mathcal M}_{\mathrm{M}}$ is the domain in $\mathscr E$ where $\Xi_{\mathrm{M}} > 0$. As in the case of the de Sitter space time, we want to identify an initial data set which evolves into Minkowski space-time $(\tilde{\mathcal M}_{\mathrm M}, \tilde{\bm g}_{\mathrm M})$. This can however not be done in an analogous way as for the de Sitter space-time  since the initial data for the conformal Einstein field equations are generically singular at the point $i^0 = \{\chi = \pi, \tau = 0\}$ describing space-like infinity ---see e.g. \cite[Chapter 21]{VKbook}. For this reason we only prove a semi-global stability result where we prescribe initial data on the hypersurface 
\begin{equation}
\bar{\mathcal S} = \left\{0 \leq \psi \leq \frac{\pi}{2}, \,\tau = \frac \pi 2\right\}
\end{equation}
which corresponds to a hyperboloid ---any other  positive value for $\tau$ would work as well. Only a portion of the Minkowski (-like) space-time can be obtained this way. Following \cite{MR868737} we construct an initial data set for the conformal Einstein field equations with massless Vlasov matter (\ref{cefmvm1})--(\ref{cefmvm2}) on $\mathbb S^3$ with help of a so called {\em extension operator} 
\begin{equation}
E: H^m (\mathcal S; \, \mathbb R^{N_g}) \to H^m (\mathbb S^3; \, \mathbb R^{N_g}).
\end{equation}
This operator permits to extend the initial data set on $\mathcal S$ to a vector-valued function  on the whole of  $\mathbb S^3$ (which does not satisfy the constraints on $\mathbb{S}^3\setminus \mathcal{S}$) such that for a constant $K > 0$ there holds for any $\mathbf{w} \in H^m(\mathcal S; \, \mathbb R^{N_g})$
\begin{equation}
(E\mathbf{w})(x) = \mathbf{w}(x) \quad \mathrm{a.e.\,in\,} \mathcal S, \quad  \| E \mathbf{w} \|_m \leq K \|\mathbf{w}\|_{H^m (\mathcal S; \, \mathbb R^{N_g})}.
\end{equation}
In view of the above we write
\begin{equation}
\mathbf{u}_{\mathrm{M}}^\star \equiv  E\left(\mathbf{u}_{\mathrm{M}} \big|_{\mathcal S}\right).
\end{equation}
Accordingly $\mathbf{u}_{\mathrm{M}}^\star$ are initial data which, if prescribed on $\mathbb S^3$, yield a solution $\mathbf{u}$ of the conformal Einstein field equations on $[0,\infty)\times \mathbb S^3$ which coincides with Minkowski space-time in the causal future $D^+(\bar{\mathcal S})$ of $\bar{\mathcal S}$. \par
As in the stability analysis of de Sitter space-time, we consider a perturbation ansatz $\bm g_{\mathrm{M}} + \breve{\bm g}$ on the Einstein cylinder $\mathscr E$ which gives rise to initial data on $\bar{\mathcal S}$. This initial data can, by virtue of the operator $E$, be extended to initial data $\mathbf{u}^\star$ on $\mathbb S^3$. It turns out (cf.~\cite{VKbook}) that this initial data is of the form $\mathbf{u}^\star = \mathbf{u}_{\mathrm{M}} + \breve{\mathbf{u}}^\star$ and that $\mathbf{u}^\star \in \bm D_{(\varepsilon, 1/2)}$, if $\breve{\bm g}$ is small enough, where in the definition (\ref{def_ded}) of $\bm D_{(\varepsilon, 1/2)}$ we choose $\mathbf{u}_{\mathrm M}^\star$ as background solution $\mathring{\mathbf{u}}^\star$. We assume that the initial data $\mathbf{u}^\star$ solves the conformal constraint equations on $\mathcal S$ with source terms  generated by the initial value $\mathbf{u}_{\mathfrak f}^\star$ of the particle distribution function. Then Theorem \ref{main_theorem} yields a solution $\mathbf{u}$ (which can be extended) on $\mathcal M_\bullet = [\pi/2, \tau_\bullet] \times \mathbb S^3 \times \mathbb R_v^4$, where $\tau_\bullet \geq 3\pi/4$, which is close to the Minkowski solution $\mathbf{u}_{\mathrm{M}}$. \par
The structure of the conformal boundary of the space-time thus obtained can be analysed using the methods used in \cite{MR868737} without any further modification. These methods allow to show that the development of the hyperboloidal initial data set is asymptotically simple. Moreover, there exists  a point $i^+ \in (\pi/2, \tau_\bullet) \times \mathbb S^3$ at which the generators of null infinity intersect. The geodesic completeness of the spacetime can be studied with similar methods to those used for perturbations of de Sitter spacetime ---see Remark \ref{Remark:GeodesicCompleteness}. 

The above observations can be collected in the following:

\begin{theorem}
Given $m\geq 5$, hyperboloidal initial data $\mathbf{u}_\star =\mathbf{u}_{\mathrm{M}}+\breve{\mathbf{u}}$ to the Einstein-Vlasov conformal constraint equations such that $|| \breve{\mathbf{u}}||_m<\varepsilon$ for $\varepsilon>0$ suitably small gives rise to a unique $C^{m-2}$ solution to the conformal Einstein-Vlasov equations such that there exists a point $i^+$ such that the causal past $J^-(i^+)$ of $i^+$ and the future domain of dependence $D^+(\bar{\mathcal S})$ of $\bar{\mathcal S}$ in the Lorentz-space $(\mathcal M_\bullet, \bm g)$ coincide. The conformal factor $\Xi$ is positive on $D^+(\bar{\mathcal S}) \backslash H^+(\bar{\mathcal S})$, where $H^+(\bar{\mathcal S})$ denotes the future Cauchy horizon of $\bar{\mathcal S}$. Furthermore, $\Xi$ vanishes on $H^+(\bar{\mathcal S})$, $\mathrm d \Xi \neq 0$ on $\mathscr I^+ = H^+(\bar{\mathcal S})  \backslash \{i^+\}$, and $\mathrm d\Xi = 0$ but the Hessian is non-degenerate at $i^+$. In particular, one has the following: the metric $\tilde{\bm g} = \Xi^{-2} \bm g$ on $\mathcal M_\bullet$ together with the particle distribution function $f$ is, whenever $\Xi\neq 0$, a solution of class $C^{m-2}$ (with curvature tensor of class $C^{m-2}$) of the massless Einstein-Vlasov system which is future asymptotically simple, thus future null geodesically complete, for which $\mathscr I^+$ represents future null infinity and $i^+$ future time-like infinity.
\end{theorem}

\printbibliography
\end{document}